\documentclass[journal,10pt]{IEEEtran}

\usepackage{cite}
\usepackage{graphicx}
\usepackage{subcaption}
\usepackage[cmex10]{amsmath}
\usepackage{algpseudocode}
\usepackage{algorithmicx} 
\usepackage{array}
\usepackage{url}
\usepackage{caption}
\usepackage{epsfig}
\usepackage{amsfonts,amssymb,multirow,bigstrut,booktabs,ctable,latexsym}

\usepackage[mathscr]{eucal}
\usepackage{verbatim}
\usepackage{amsthm}
\usepackage{algorithm}

\begin{document}

\newtheorem{definition}{Definition}
\newtheorem{lemma}{Lemma}
\newtheorem{theorem}{Theorem}
\newtheorem{example}{Example}
\newtheorem{proposition}{Proposition}
\newtheorem{remark}{Remark}
\newtheorem{assumption}{Assumption}
\newtheorem{corrolary}{Corrolary}
\newtheorem{property}{Property}
\newtheorem{ex}{EX}
\newtheorem{problem}{Problem}
\newcommand{\argmin}{\arg\!\min}
\newcommand{\argmax}{\arg\!\max}
\newcommand{\st}{\text{s.t.}}
\newcommand \dd[1]  { \,\textrm d{#1}  }
%
\title{\LARGE \bf A Differentially Private Incentive Design for Traffic Offload to Public Transportation}

\author{Luyao Niu,~\IEEEmembership{Student Member,~IEEE,} and Andrew Clark,~\IEEEmembership{Member,~IEEE}%
\thanks{L. Niu and A. Clark are with the Department of Electrical and Computer Engineering, Worcester Polytechnic Institute, Worcester, MA 01609 USA.
{\tt\small \{lniu,aclark\}@wpi.edu}}
}

\maketitle
\thispagestyle{empty}
\pagestyle{empty}

\begin{abstract}
Increasingly large trip demands have strained urban transportation capacity, which consequently leads to traffic congestion and rapid growth of greenhouse gas emissions. In this work, we focus on achieving sustainable transportation by incentivizing passengers to switch from private cars to public transport. We address the following challenges. First, the passengers incur inconvenience costs when changing their transit behaviors due to delay and discomfort, and thus need to be reimbursed. Second, the inconvenience cost, however, is unknown to the government when choosing the incentives. Furthermore, changing transit behaviors raises privacy concerns from passengers. An adversary could infer personal information, (e.g., daily routine, region of interest, and wealth), by observing the decisions made by the government, which are known to the public. We adopt the concept of differential privacy and propose privacy-preserving incentive designs under two settings, denoted as two-way communication and one-way communication. Under two-way communication, passengers submit bids and then the government determines the incentives, whereas in one-way communication the government simply sets a price without acquiring information from the passengers. We formulate the problem under two-way communication as a mixed integer linear program, and propose a polynomial-time approximation algorithm. We show the proposed approach achieves truthfulness, individual rationality, social optimality, and differential privacy. Under one-way communication, we focus on how the government should design the incentives without revealing passengers' inconvenience costs while still preserving differential privacy. We formulate the problem as a convex program, and propose a differentially private and near-optimal solution algorithm. A numerical case study using Caltrans Performance Measurement System (PeMS) data source is presented as evaluation. The results show that the proposed approaches achieve a win-win situation in which both the government and passengers obtain non-negative utilities.

\end{abstract}

\section{Introduction}

Rapid urbanization is a global trend \cite{urbanization}. Compared to public and non-motorized transport modes, private vehicles are an increasingly popular transport choice to meet the huge traffic demands associated with the fast growing urban population \cite{dimitriou2011urban}. It has been shown that around $47\%$ of daily trips in cities are made by private motorized vehicles \cite{pourbaix2011towards}. If such trends continue, it is predicted that there will be $6.2$ billion daily trips made by private vehicles in 2025 \cite{pourbaix2011towards}.

Several challenges are raised due to the fast growing trip demands and increasingly pervasive uses of private vehicles. First, road transport is overly consumed, resulting in traffic congestion which leads to economic losses. For example, the cost caused by congestion in urban areas in USA in 2010 is approximately $\$101$ billion \cite{2010urban}. The European Union (EU) estimates the cost incurred due to congestion to be $1\%$ of its annual gross domestic product (GDP) \cite{EU2011}. Second, environmental concerns are raised due to the growth of private car use. The Greenhouse Gas (GHG) emissions due to transportation sector will reach $40\%$ by 2050 \cite{IEA}. 

Reducing the dependence on private cars has been identified as one of the objectives of governments to achieve sustainability \cite{Benefits}. One approach is to promote public transportation, which is shown to be more sustainable compared to private cares \cite{holmgren2007meta,cervero2013transport}, as an alternative \cite{handy2005driving}.

In this work, we investigate the problem of how the government could incentivize the passengers to use public transport instead of private cars. There are several challenges faced by the government to encourage passengers changing transit behaviors--from private cars to public transport. First, although the government discourages the use of private cars, passengers' trip demands still need to be satisfied. Moreover, the passengers incur inconvenience costs when switching from private cars to public transport. The inconvenience cost is due to several factors including reduced quality of service (QoS) and delay of arrival time. The passengers need to be reimbursed for these costs. Furthermore, the inconvenience cost, which varies from passenger to passenger, is unknown to the government. The passengers might be unwilling to reveal the inconvenience costs, or lie on their inconvenience costs to earn benefits during the interaction with the government. Although existing literature has experimentally identified the factors that prevent passengers from changing transit behavior \cite{bhat2006impact,jain2014identifying,redman2013quality,wall2017encouraging}, a theoretical analysis on how to incentivize the passengers to change their transit behaviors with privacy guarantee has received little research attention.

In this paper, we model and analyze how the government could incentivize the passengers to satisfy their traffic demands via public transport instead of private cars under two settings, named two-way communication and one-way communication. Under two-way communication, the government and passengers can communicate with each other. Under one-way communication, the government can send information to the passengers but not vice versa. We not only address the challenges faced by the government, but also address the passengers' privacy concerns when shifting from private cars to public transport. The privacy concerns are raised since an untruthful party can observe how the passengers respond to incentives and learn the passengers' private information including region of interest and daily routine. Such privacy concerns discourage privacy sensitive passengers to switch from private cars to public transport. We make the following contributions.

\begin{itemize}
    \item We model the interaction between the government and passenger under two-way communication using a reverse auction model. We formulate the problem of incentivizing the passengers as a mixed integer linear program. We propose an efficient algorithm to reduce the computation complexity for computing the passengers selected by the government and the associated incentives. 
    \item We prove that the proposed mechanism design under two-way communication achieves approximate optimal social welfare, truthfulness, individual rationality, and differential privacy.
    \item For the one-way communication setting, we formulate the problem as an online convex program. We give a polynomial-time algorithm to solve for the mechanism design. We prove that the proposed mechanism is differentially private and provides the same asymptotic utility as the best fixed price, i.e., achieves Hanan consistency.
    \item We present a numerical case study with real-world trace data as evaluation. The results show that the proposed approach achieves individual rationality and non-negative social welfare, and is privacy preserving.
\end{itemize}

The remainder of this paper is organized as follows. We discuss the related works in Section \ref{sec:related}. In Section \ref{sec:formulation}, we present the problem formulation under two-way and one-way communication settings, respectively. We present the proposed incentive mechanism design in Section \ref{sec:auction} for the two-way communication setting. Section \ref{sec:one way} gives the proposed solution for one-way communication setting. The proposed approaches are demonstrated using a numerical case study in Section \ref{sec:simulation}. We conclude the paper in Section \ref{sec:conclusion}.

\section{Related Work}\label{sec:related}
In this section, we present literature review on intelligent transportation systems and differential privacy. Significant research effort has been devoted to achieving intelligent and sustainable transportation systems. Planning and routing navigation problems have been investigated by transportation and control communities \cite{lam2016learning,vasile2017minimum,como2013robust1,como2013robust2,pavone2012robotic,zhang2016model}. Various approaches have been proposed to improve operation efficiency of existing transportation infrastructure, among which vehicle balancing \cite{miao2017data} has been extensively studied for bike sharing \cite{schuijbroek2017inventory} and taxis \cite{miao2016taxi}. Metering strategies have also been investigated \cite{coogan2014freeway}. Different from the works mentioned above, this paper focuses on the demand side management, with particular interest on how to encourage passengers to change their transit behaviors via incentive design. 

In the following, we discuss related works on demand side management. Alternative travel infrastructures such as bike sharing system \cite{demaio2009bike}, have been implemented all over the world. Moreover, ridesharing system and the associated ridesharing match system have been investigated \cite{kleiner2011mechanism,ma2013t}, which grouped passengers with similar itineraries and time schedules together to reduce the number of operating vehicles. Most of these works focus on taxis and ride-hailing services such as Uber and Lyft, and ignore the potential from public transport. Pricing schemes have been proposed to reduce the number of operating vehicles at peak hour \cite{goh2002congestion,kachroo2017optimal}. These works focus on private cars and ignore public transport. Researchers have identified the factors (e.g., passengers' attitude and government's policy) that prevent passengers from taking public transport \cite{bhat2006impact,jain2014identifying,redman2013quality,wall2017encouraging,nurdden2007effect,beirao2007understanding}. However, to the best of our knowledge, there has been little research attention on how to design incentives to encourage passengers to switch from private to public transit services.

Mechanism design has recently been used in engineering applications such as cloud computing. In particular, Vickrey-Clarke-Groves (VCG) mechanism \cite{vickrey1961counterspeculation} is widely used to preserve truthfulness. However, truthful communication raises privacy concerns. To address the privacy issue, we adopt the concept of differential privacy \cite{dwork2008differential,dwork2006calibrating,dwork2014algorithmic}. Mechanism designs with differential privacy, such as exponential mechanism, have been proposed \cite{mcsherry2007mechanism,huang2012exponential,dwork2014algorithmic,han2018privacy}. However, they are not readily applicable to the problem investigated in this paper because the presence of inconvenience cost functions leads to violations of individual rationality. Moreover, the exponential mechanism is computationally complex, and hence in this paper we propose efficient approximation algorithms.

Trial-and-error implementation for toll pricing has been proposed in \cite{zhao2006line}. Different from \cite{zhao2006line}, we consider a closed-loop Stackelberg information pattern, and compute the optimal incentive price. To solve the problem under one-way communication setting, we adopt the Laplace mechanism to preserve differential privacy \cite{dwork2014algorithmic}. This paper extends our preliminary conference version \cite{Niu2019Differential}, in which two-way communication setting is studied. We extend the preliminary work by also investigating the one-way communication setting.
\section{Problem Formulation}\label{sec:formulation}

In this section, we first give the problem overview. Then we present the problem formulations under two settings, denoted as two-way communication and one-way communication. We finally discuss the privacy model.

\subsection{Problem Overview}

Let $\mathcal{S}=\{1,\ldots,S\}$ denote the set of origin-destination (OD) pairs that will require passengers to switch to public transport over time horizon $t=1,\ldots,T$. When passengers switch from private cars to public transport, they can provide some amount of traffic offload\footnote{In the remainder of this article, we use traffic offload and reduce the use of private car interchangeably.}. We assume each OD pair $s\in\mathcal{S}$ requires $Q_{s,t}$ amount of traffic offload at time $t$ to achieve sustainability. Let $\mathcal{N}=\{1,\ldots,N\}$ be the set of passengers. At each time $t$, any passenger $i\in\mathcal{N}$ that switches to public transport for any OD pair $s\in\mathcal{S}$ receives revenue $r_{i,s,t}(q_{i,s,t})$ issued by the government, where $q_{i,s,t}\geq 0$ is the amount of traffic offload that passenger $i$ can provide for OD pair $s$ at time $t$ and $r_{i,s,t}(0)=0$. Passenger $i$ also incurs inconvenience cost $C_{i,s}(q_{i,s,t})$ if it switches from private to public transit service due to discomfort and time of arrival delays. We remark that each passenger is physically located close to some OD pair $s$ at each time $t$. Hence each passenger is only willing to switch to public transport for one OD pair $s$ that is physically close to its current location. For other OD pairs $s'\neq s$, we can regard the associated inconvenience cost as infinite. We assume that the inconvenience cost function $C_{i,s}(q_{i,s,t})$ is continuously differentiable, strictly increasing with respect to $q_{i,s,t}$ for all $s\in\mathcal{S}$, and convex with $C_{i,s}(0)=0$ for all $i$ and $s$. The utility of the passenger at each time step $t$ is given by
\begin{equation}\label{eq:passenger utility}
    U_{i,t}=\sum_{s}\left[r_{i,s,t}(q_{i,s,t})-C_{i,s}(q_{i,s,t})\right].
\end{equation}
In this work, we assume the passengers are \emph{selfish} and \emph{rational}, i.e., the passengers selfishly maximize their utilities.


\subsection{Case 1: Interaction with Two-way Communication}\label{subsec:two way}
In this subsection, we present the problem formulation under two-way communication. In this case, the interaction between the government and the set of passengers is captured by a reverse auction model. 

The passengers act as the bidders. Each passenger can submit a bid $\mathbf{b}_{i,t}=[b_{i,1,t},\cdots,b_{i,S,t}]$ to the government at each time $t$, where element $b_{i,s,t}=(\zeta q_{i,s,t},\bar{C}_{i,s}(q_{i,s,t}))$ contains the amount of traffic offload that passenger $i$ can provide and the associated inconvenience cost. Here $\zeta$ converts the amount of traffic offload from utilities in dollars. Without loss of generality, we assume $\zeta=1$ in the remainder of this article. Note that $\bar{C}_{i,s}(q_{i,s,t})$ is the inconvenience cost claimed by passenger $i$, which does not necessarily equal the true cost $C_{i,s}(q_{i,s,t})$.

The government is the auctioneer. It collects the bids from all passengers, and then selects a set of passengers that should participate in traffic offload. In particular, the government computes a selection profile $X\in\{0,1\}^{N\times S\times T}$, with each element $x_{i,s,t}=1$ if passenger $i$ is selected and $0$ otherwise.
If a passenger $i$ is selected by the government for OD pair $s$, an associated incentive $r_{i,s,t}(q_{i,s,t})$ is issued to passenger $i$.

The utility \eqref{eq:passenger utility} of each passenger $i$ at time $t$ is rewritten as
\begin{equation}\label{eq:passenger utility 1}
     U_{i,t}=\sum_{s}x_{i,s,t}\left[r_{i,s,t}(q_{i,s,t})-C_{i,s}(q_{i,s,t})\right],~\forall i,t.
\end{equation}
The social welfare can be represented as 
\begin{equation}
    \Omega(X,B)=\sum_t\sum_s\sum_i\left[x_{i,s,t}(q_{i,s,t}-C_{i,s}(q_{i,s,t}))\right],
\end{equation}
where $B$ contains $\mathbf{b}_{i,t}$ for all $i$ and $t$. The government aims at maximizing social welfare $\Omega(X,B)$. This social welfare maximization problem is given as
\begin{subequations}\label{eq:social welfare maximization 1}
\begin{align}
    \max_{X}  \quad&\sum_t\sum_s\sum_i\left[x_{i,s,t}(q_{i,s,t}-C_{i,s}(q_{i,s,t}))\right]\\
    \text{s.t.} \quad& \sum_{s\in\mathcal{S}}x_{i,s,t}\leq 1,~\forall i,t\label{eq:social welfare maximization constr 1}\\
    &\sum_{i\in\mathcal{N}}x_{i,s,t}q_{i,s,t}\geq Q_{s,t},~\forall s,t\label{eq:social welfare maximization constr 2}\\
    &x_{i,s,t}\in\{0,1\},~\forall i,s,t\label{eq:social welfare maximization constr 3}
\end{align}
\end{subequations}
Constraint \eqref{eq:social welfare maximization constr 1} implies that a passenger can only be selected for one OD pair at each time $t$. Constraint \eqref{eq:social welfare maximization constr 2} requires the desired traffic offload $Q_{s,t}$ must be satisfied for all $s$ and $t$. Constraint \eqref{eq:social welfare maximization constr 3} defines binary variable $x_{i,s,t}$.

Under the two-way communication setting, a malicious adversary aims at inferring the inconvenience cost function of each passenger by observing the selection profile $X$. The adversary can observe $X$ be eavesdropping on communication channel. Let $X_t$ be the selection profile at time $t$. Then the information perceived by the adversary up to time $t$ is $\mathcal{I}_t^{two}=\{X_{t'}|t'\leq t\}$. In this case, the government needs to compute a privacy preserving incentive mechanism such that the passengers truthfully report their inconvenience cost functions so that the social welfare is (approximately) optimal. 

Besides the privacy guarantees, we state some additional desired properties that the government needs to achieve under this two-way communication setting. First, individual rationality for each passenger should be achieved, i.e., each passenger must obtain non-negative utility when being selected by the government. Second, the government wishes to reveal the true inconvenience cost functions from the passengers to seek the optimal solution to \eqref{eq:social welfare maximization 1}. Therefore the government needs to ensure that the passengers bid truthfully. Truthfulness is defined as follows.
\begin{definition}
\emph{(Truthfulness).} An auction is truthful if and only if bidding the true inconvenience cost function, i.e., $\bar{C}_{i,s}(q_{i,s,t})=C_{i,s}(q_{i,s,t})$ for all $q_{i,s,t}$, is the dominant strategy for any passenger $i$ regardless of the bids from the other passengers. In other words, bidding $\bar{C}_{i,s}(q_{i,s,t})=C_{i,s}(q_{i,s,t})$ maximizes the utility \eqref{eq:passenger utility 1} of passenger $i$ for all $i$.
\end{definition}

\subsection{Case 2: Interaction with One-way Communication}\label{subsec:one way}

In this subsection, we present a problem formulation when two-way communication is infeasible, while one-way communication from the government to the passengers is enabled. Under this setting, the passengers cannot report any information to the government. The government hence broadcasts an incentive price $p_{s,t}$ for each OD pair $s$ at each time step $t$, and then observes the responses from the passengers to design the incentive price for next time step $(t+1)$. Different from two-way communication, the passengers respond to the incentive price rather than bidding a fixed amount of traffic offload. Hence, the amount of traffic offload provided by each passenger $i$ for OD pair $s$ at time $t$ is defined as a function of incentive price $p_{s,t}$, denoted as $q_{i,s}(p_{s,t})$. We assume that the the traffic offload $q_{i,s}(p_{s,t})$ provided by each passenger $i$ is strictly increasing with respect to $p_{s,t}$. 

The government predicts the traffic condition for the set of OD pairs $\mathcal{S}=\{1,2,\cdots,S\}$ in the near future time horizon $t=1,\cdots,T$ based on the historical traffic information (e.g., traffic conditions during rush hours). Suppose the government requires $Q_{s,t}\geq 0$ amount of traffic offload on OD pair $s$ at each time index $t$. To satisfy $Q_{s,t}$ amount of traffic offload, the government designs a unit incentive price $p_{s,t}$ for each time index $t$ to incentivize individual passengers to participate in the traffic offload program. The information perceived by the government $\mathcal{I}_{t}^{gov}$ up to time $t$ includes the following: (i) the historical incentives $\{p_{s,t'}|t'=1,\cdots,t-1,s\in\mathcal{S}\}$, (ii) the historical traffic offload offered by the passengers $\{q_{i,s}(p_{s,t'})|i\in\mathcal{N},s\in\mathcal{S},t'=1,\cdots,t-1\}$. Thus the government's decision on $p_{s,t}$ for each time $t$ and OD pair $s$ can be interpreted as a policy mapping from the information set to the set of non-negative real numbers $p_{s,t}:\mathcal{I}_t^{gov}\mapsto\mathbb{R}_{\geq 0}$.

At each time step $t$, the passengers observe the incentives $p_{s,t}$, and then decide whether to participate in traffic offload and earn the incentive $p_{s,t}q_{i,s}(p_{s,t})$ based on their own utility functions. Passengers that participate in traffic offload incur inconvenience cost $C_{i,s}(q_{i,s}(p_{s,t}))$. The inconvenience cost function $C_{i,s}(q_{i,s}(p_{s,t}))$ is private to each passenger $i$. The information $\mathcal{I}_t^i$ available to passenger $i$ up to time $t$ includes the following: (i) the historical incentives $\{p_{s,t'}|t'=1,\cdots,t,s\in\mathcal{S}\}$, (ii) the traffic offload function $\{q_{i,s}(\cdot)|s\in\mathcal{S}\}$, and (iii) its inconvenience cost function $\{C_{i,s}(\cdot)|s\in\mathcal{S}\}$. 

Let $\mathbf{p}_t=[p_{1,t},\cdots,p_{S,t}]$ be the incentive prices for all OD pairs $s\in\mathcal{S}$ at time $t$. The utility of each passenger $i$ at time step $t$ can be represented as
\begin{equation}
    U_{i,t}(\mathbf{p}_t)=\sum_{s\in\mathcal{S}}\left\{p_{s,t}q_{i,s}(p_{s,t})-C_{i,s}(q_{i,s}(p_{s,t}))\right\},~\forall i,t.
\end{equation}
The social cost is given by
\begin{multline}\label{eq:gov utility}
    \Lambda(\mathbf{p})=\sum_t\sum_{s\in\mathcal{S}}\Bigg\{\sum_{i\in\mathcal{N}}C_{i,s}(q_{i,s}(p_{s,t}))\\+\beta_s\left[Q_{s,t}-\sum_{i\in\mathcal{N}}q_{i,s}(p_{s,t}) \right]^+\Bigg\},
\end{multline}
where $\mathbf{p}=[\mathbf{p}_1,\cdots,\mathbf{p}_T]^T$ contains the incentive prices for all $s$ and $t$, $[\cdot]^+$ represents $\max\{\cdot,0\}$, and $\beta_s$ represents the penalty due to deficit of traffic offload. The social cost minimization problem is formulated as $\min_{\mathbf{p}} \Lambda(\mathbf{p}).$

Under the one-way communication setting, the malicious party could not observe the participation of each passenger directly as in two-way communication setting. We focus on a malicious party that can observe the incentive prices issued by the government up to time $t$ and then infer the amount of traffic offload offered by each passenger $i$, which might be further used to infer the inconvenience cost functions of the passengers. Denote the information obtained by the government up to time $t$ as $\mathcal{I}_{t}^{one}$. Then we have $\mathcal{I}_{t}^{one}=\{p_{s,t'}|\forall s,\forall t'\leq t\}$. The objective of a malicious party is to compute $q_{i,s}(p_{s,t'})$ given $\mathcal{I}_{t}^{mal}$. In this case, the government's objective is to compute a privacy preserving incentive design such that the social welfare is (approximately) maximized.

Besides the privacy guarantee, we briefly discuss the game-theoretic properties under one-way communication setting. Since the government broadcasts the incentive price while the passengers decide if they will participate or not, individual rationality is automatically guaranteed for rational passengers. Truthfulness is not required under one-way communication setting since the passengers cannot send messages to the government under this setting.

\subsection{Notion of Privacy}

In this subsection, we give the notion of privacy adopted in this paper. We focus on differential privacy \cite{dwork2008differential,dwork2006calibrating}, which is defined as follows.
\begin{definition}\label{def:differential privacy}
\emph{($\epsilon$-Differential Privacy.)} Given $\epsilon\geq 0$, a computation procedure $M$ is said to be $\epsilon$-differentially private if for any two inputs $C_1$ and $C_2$ that differ in a single element and for any set of outcomes $L\subseteq\text{Range}(M)$, the relationship $Pr(M(C_1)\in L)\leq \exp(\epsilon)\cdot Pr(M(C_2)\in L)$ holds, where $\text{Range}(M)$ is the set of all outcomes of $M$.
\end{definition}
Definition \ref{def:differential privacy} requires computation procedure $M$ to behave similarly given similar inputs, where parameter $\epsilon$ models how similarly the procedure should behave. A more relaxed and general definition of differential privacy is as follows.
\begin{definition}\label{def:relaxed differential privacy}
\emph{($(\epsilon,\delta)$-Differential Privacy.)} Given $\epsilon\geq 0$ and $\delta\geq 0$, a computation procedure $M$ is said to be $(\epsilon,\delta)$-differentially private if for any two inputs $C_1$ and $C_2$ that differ in a single element and for any set of outcomes $L\subseteq\text{Range}(M)$, inequality $Pr(M(C_1)\in L)\leq \exp(\epsilon)\cdot Pr(M(C_2)\in L)+\delta$ holds.
\end{definition}

To quantify the privacy leakage using the proposed incentive designs, we adopt the concept of min-entropy leakage \cite{barthe2011information}. We first introduce the concepts of min-entropy and conditional min-entropy \cite{renyi1961measures}, and then define the min-entropy leakage. Let $V$ and $Y$ be random variables. The min-entropy of $V$ is defined as $H_{\infty}(V)=\lim_{\alpha\rightarrow\infty}\frac{1}{1-\alpha}\log_2\sum_{v}Pr(V=v)^\alpha$, where $Pr(V=v)$ represents the probability of $V=v$. The conditional min-entropy is defined as
$H_{\infty}(V|Y)=-\log_2\sum_{y}Pr(Y=y)\max_vPr(v|y)$, where $Pr(v|y)$ is the probability that $V=v$ given that $Y=y$. Then the min-entropy leakage \cite{barthe2011information} is defined as $L=H_{\infty}(V)-H_{\infty}(V|Y).$

Under two-way communication setting, the min-entropy leakage is computed as
\begin{multline*}
    L=\lim_{\alpha\rightarrow\infty}\frac{1}{1-\alpha}\log_2\sum_{B}Pr(B)^\alpha-\\
    \left(-\log_2\sum_{X}Pr(X)\max_BPr(B|X)\right),
\end{multline*} 
where $Pr(B)$ is the probability that a bidding profile $B$ is submitted, and $Pr(B|X)$ is the probability that the bidding profile $B$ is submitted given the selection profile $X$ is observed. Under one-way communication setting, the min-entropy leakage is computed as
\begin{multline*}
    L=\lim_{\alpha\rightarrow\infty}\frac{1}{1-\alpha}\log_2\sum_{C}Pr(C)^\alpha-\\
        \left(-\log_2\sum_{\mathbf{p}}Pr(\mathbf{p})\max_CPr(C|\mathbf{p})\right),
\end{multline*} 
where $Pr(C)$ is the probability that the collection of passengers' inconvenience cost functions is $C$, and $Pr(C|\mathbf{p})$ is the probability that the collection of inconvenience costs is $C$ given the historical incentives $\mathbf{p}$ is observed.

\section{Solution for Two-way Communication Setting}\label{sec:auction}
Motivated by exponential mechanism \cite{huang2012exponential,mcsherry2007mechanism}, we present an incentive design for the two-way communication setting in this section. We propose a payment scheme that achieves individual rationality. We mitigate the computation complexity incurred in exponential mechanism using an iterative algorithm. We prove that the desired properties are achieved using the proposed incentive design.

\subsection{Solution Approach}

In this subsection, we give an exact solution under two-way communication. We formally prove that truthfulness, approximate social welfare maximizing, and differential privacy are achieved using the proposed mechanism.

The mechanism is presented in Algorithm \ref{algo:mechanism}. The algorithm takes the bid profile from the passengers as input, and gives the selection profile $X$ and the incentives issued to each selected passenger. The algorithm works as follows. At each time $t\leq T$, the government selects a feasible solution to social welfare maximization problem \eqref{eq:social welfare maximization 1}. The probability of selecting each feasible $X$ is proportional to the exponential function evaluated at the associated social welfare $\Omega(X,B)$ with scale $\frac{\epsilon}{2\Delta}$, where $\Delta$ is the difference between the upper and lower bound of social welfare $\Omega(X,B)$. Although the computation of selection profile $X$ is motivated by exponential mechanism \cite{huang2012exponential,mcsherry2007mechanism}, the Vickrey-Clarke-Groves (VCG)-like payment scheme adopted by exponential mechanism is not applicable to the problem investigated in this work. The reason is that the VCG-like payment scheme violated individual rationality and truthfulness in our case, due to the fact that the passengers do not only have valuations over the incentives, but also inconvenience costs during traffic offload. To this end, the payment scheme \eqref{eq:incentive 1} is proposed for the problem of interest, in which the incentive issued to each passenger is determined by the social cost introduced by each passenger. In the following, we characterize the mechanism presented in Algorithm \ref{algo:mechanism}. 
\begin{theorem}\label{thm:properties of sol 1}
The mechanism described in Algorithm \ref{algo:mechanism} achieves truthfulness, individual rationality, near optimal social welfare, and $\epsilon$-differential privacy.
\end{theorem}
\begin{proof}
We omit the proof due to space limit. See \cite{Niu2019Differential} for a detailed proof.
\end{proof}

The mechanism proposed in Algorithm \ref{algo:mechanism} is computationally expensive. The payment scheme \eqref{eq:incentive 1} is intractable when the passenger set is large since \eqref{eq:incentive 1} needs to compute the social welfare associated with $X$ and $X_{-i}$ for all $i$. Therefore, a computationally efficient algorithm is desired.

\begin{center}
  	\begin{algorithm}[!htp]
  		\caption{Mechanism design for the government.}
  		\label{algo:mechanism}
  		\begin{algorithmic}[1]
  			\Procedure{Mechanism}{$B$}
  			\State \textbf{Input}: Bid profile $B$
  			\State \textbf{Output:} Selection profile $X$, incentives $R$
  			\While{$t\leq T$}
  			\State Choose a selection profile $X$ that is feasible for social welfare maximization problem \eqref{eq:social welfare maximization 1} with probability
    \begin{equation}\label{eq:selection}
        Pr(X)\propto \exp\left(\frac{\epsilon}{2\Delta}\Omega(X,B)\right).
    \end{equation}
  			\State For each passenger that is selected, issue incentive $r_{i}$ as
    \begin{multline}\label{eq:incentive 1}
        r_{i,s,t}=\underset{X\sim D(\mathbf{b}_{i,t},B_{-i,t})}{\mathbb{E}}\bigg\{\sum_{j}\sum_sx_{j,s,t}q_{j,s,t}\\-\sum_{j'\neq i}\sum_sx_{j',s,t}C_{j',s}(q_{j',s,t})\bigg\}
        +\frac{2\Delta}{\epsilon}E(D(\mathbf{b}_{i,t},B_{-i,t}))\\
        -\frac{2\Delta}{\epsilon}\ln\left(\sum_{X}\exp\left(\frac{\epsilon}{2\Delta}\Omega(X_{-i},B_{-i})\right)\right),
    \end{multline}
    where $\Delta$ is the difference between the upper and lower bound of social welfare $\Omega(X,B)$, $E(\cdot)$ is the Shannon entropy, $D(\cdot)$ is the probability distribution over selection profile $B$, and $X_{-i,t}$ and $B_{-i,t}$ are the matrix obtained by removing the $i$-th row and $i$-th column in selection profile and bid profile, respectively.
    \State $t\leftarrow t+1$
    \EndWhile
  			\EndProcedure
  		\end{algorithmic}
  	\end{algorithm}
  \end{center}

\subsection{Efficient Algorithm}

Algorithm \ref{algo:mechanism} is computationally intensive and hence we need an efficient algorithm. In this subsection, we give a mechanism that achieves the desired game-theoretic properties and privacy guarantees and runs in polynomial time.

In real world implementation, since the passengers are geographically distributed, the government can decompose the social welfare maximization problem \eqref{eq:social welfare maximization 1} with respect to OD pair $s$. Then problem \eqref{eq:social welfare maximization 1} becomes a set of optimization problems associated with each OD pair $s$ as follows:
\begin{align}
    \max_{\mathbf{x}_s}  \quad&\sum_{i\in\mathcal{N}}x_{i,s,t}\left(q_{i,s}-C_{i,s}(q_{i,s})\right)\label{eq:decomposed social welfare maximization}\\
    \text{s.t.} \quad& 
    \sum_{i\in\mathcal{N}}x_{i,s,t}q_{i,s,t}\geq Q_{s,t},~\forall s,t\nonumber\\
    &x_{i,s,t}\in\{0,1\},~\forall i,s,t.\nonumber
\end{align}

Given the set of decomposed problems, if we can achieve the optimal solution to each decomposed problem using an incentive design, then we reach social optimal solution. Thus our objective is design a mechanism that achieves the (approximate) optimal solution of each decomposed problem, individual rationality, truthfulness, and differential privacy.

The proposed efficient algorithm for each decomposed problem is presented in Algorithm \ref{algo:subproblem}. The algorithm iteratively computes the set of passengers $\mathcal{W}_{s,t}$ selected by the government for OD pair $s$ at time $t$. First, the set $\mathcal{W}_{s,t}$ is initialized as an empty set. Then at each iteration $k$, the probability that selecting a passenger $i$ that has not been selected at time $t$ is proportional to the exponential function $\exp\left(\epsilon'(q_{i,s,t}-\bar{C}_{i,s}(q_{i,s,t}))\right)$, i.e.,
\begin{multline}\label{eq:selection probability}
    Pr\left(\mathcal{W}_{s,t}\leftarrow\mathcal{W}_{s,t}\cup\{i\}\right)\\\propto
    \begin{cases}
    \exp\left(\epsilon'(q_{i,s,t}-\bar{C}_{i,s}(q_{i,s,t}))\right),
    &\mbox{if $i$ has not been selected;}\\
    0&\mbox{otherwise;}
    \end{cases}
\end{multline}
where $\epsilon'=\frac{\epsilon}{e\ln(e/\delta)}$. Then the set of selected passengers $\mathcal{W}_{s,t}$ are removed from the passenger set $\mathcal{N}$. For each $i\in\mathcal{W}_{s,t}$, the government issues incentive $r_{i,s,t}$ computed as
\begin{multline}\label{eq:payment design}
    r_{i,s,t}=(q_{i,s,t}+z)\exp\left(\epsilon'(q_{i,s,t}-\bar{C}_{i,s}(q_{i,s,t}))\right)\\
    -\int_0^{q_{i,s,t}+z}\exp(\epsilon'y)\text{d}y,
\end{multline}
where $z=\frac{\bar{C}_{i,s}(q_{i,s,t})}{\exp\left(\epsilon'(q_{i,s,t}-\bar{C}_{i,s}(q_{i,s,t}))\right)}$.
We characterize the solution presented in Algorithm \ref{algo:subproblem} as follows.
\begin{lemma}
Algorithm \ref{algo:subproblem} achieves truthfulness, individual rationality, and $\left(\frac{\epsilon\Delta}{e(e-1)},\delta\right)$-differential privacy. Moreover, Algorithm \ref{algo:subproblem} achieves near optimal social welfare $\Omega_s^\ast-O(\ln Q_s)$ with probability at least $1-\frac{1}{{Q_s}^{O(1)}}$, where $\Omega_s^\ast$ is the maximum social welfare for OD pair $s$.
\end{lemma}
\begin{proof}
We omit the proof due to space limit. See \cite{Niu2019Differential} for detailed proof.
\end{proof}

Given Algorithm \ref{algo:subproblem} for each decomposed problem, we present Algorithm \ref{algo:solution}, which utilizes Algorithm \ref{algo:subproblem} as subroutine, to solve for the selection profile $X$ for problem \eqref{eq:social welfare maximization 1}. Algorithm \ref{algo:solution} works as follows. It first makes $S$ copies of the passenger set $\mathcal{N}$, with each denoted as $\mathcal{N}_s$ for all $s\in\mathcal{S}$. Then Algorithm \ref{algo:subproblem} is invoked iteratively to compute the selected passengers for each OD pair $s$. The selection profile $X$ for time $t$ is finally returned as the union $\cup_{s}\mathcal{W}_{s,t}$.

\begin{center}
  	\begin{algorithm}[!htp]
  		\caption{Solution algorithm for decomposed problem \eqref{eq:decomposed social welfare maximization}.}
  		\label{algo:subproblem}
  		\begin{algorithmic}[1]
  			\Procedure{Decompose}{$B$}
  			\State \textbf{Input}: Bid profile $B$, current time $t$
  			\State \textbf{Output:} Selection profile $\mathcal{W}_{s,t}$
  			\State \textbf{Initialization}: Selected passenger set $\mathcal{W}_{s,t}\leftarrow\emptyset$, $\epsilon'\leftarrow\frac{\epsilon}{e\ln(e/\delta)}$
            \While{$|\mathcal{W}_{s,t}|\leq Q_s\land\mathcal{N}\neq\emptyset$}
            \For {$i\in\mathcal{N}$}
            \State Compute the probability of selecting passenger $i$ as \eqref{eq:selection probability}.
            \EndFor
            \If{passenger $i$ is chosen}
            \State $\mathcal{N}\leftarrow\mathcal{N}\setminus\{i\}$
            \EndIf
            \EndWhile
            \State \Return $\mathcal{W}_{s,t}$
  			\EndProcedure
  		\end{algorithmic}
  	\end{algorithm}
  \end{center}

  \begin{center}
  	\begin{algorithm}[!htp]
  		\caption{Solution algorithm for problem \eqref{eq:social welfare maximization 1}}
  		\label{algo:solution}
  		\begin{algorithmic}[1]
  			\Procedure{Social\_Max}{$B$}
  			\State \textbf{Input}: Bid profile $B$
  			\State \textbf{Output:} Selection profile $X$
  			\While{$t\leq T$}
  			\State \textbf{Initialization}: $\mathcal{N}_s=\mathcal{N}$ for all $s$
  			\State Remove all passengers that provide negative social welfare $B\leftarrow [(q_{i,s},\bar{C}_{i,s}):q_{i,s,t},\bar{C}_{i,s}(q_{i,s,t})\geq 0]$
            \For{$s\in\mathcal{S}$}
            \State \textproc{Decompose($B$)}
            \State $\mathcal{N}_s=\mathcal{N}_s\setminus\cup_{s'=1}^{s-1}\mathcal{W}_{s'}$
            \EndFor
            \State \Return $X=\cup_{s\in\mathcal{S}}\mathcal{W}_{s,t}$
            \State $t\leftarrow t+1$
            \EndWhile
  			\EndProcedure
  		\end{algorithmic}
  	\end{algorithm}
  \end{center}

We conclude this section by characterizing the properties achieved by Algorithm \ref{algo:solution}.

\begin{theorem}\label{thm:two-way property(efficient)}
Algorithm \ref{algo:solution} achieves truthfulness, individual rationality, and $\left(\frac{\epsilon\Delta S}{e(e-1)},\delta S\right)$-differential privacy. Moreover, Algorithm \ref{algo:subproblem} achieves near optimal social welfare $\Omega^\ast-SO(\ln Q_s)$ with at least probability $1-\frac{1}{{Q^*}^{O(1)}}$, where $\Omega^\ast$ is the maximum social welfare and $Q^*=\max_sQ_s$.
\end{theorem}
\begin{proof}
We omit the proof due to space limit. See \cite{Niu2019Differential} for detailed proof.
\end{proof}
\section{Solution for One-way Communication Setting}\label{sec:one way}

In this section, we analyze the problem formulated in Section \ref{subsec:one way}. We first present an incentive mechanism design without privacy guarantee. Then we generalize the analysis and give an incentive design that satisfies differential privacy. 

\subsection{Incentive Mechanism Design without Privacy Guarantee}\label{subsec:price without privacy}
Different from the two-way communication scenario, the passengers observe the incentive price signal sent by the government and respond to it by maximizing their own utility. In the following, we first analyze passengers' best responses to price signal. Then we analyze how the government should design the incentive price to achieve optimal social welfare.
\begin{lemma}\label{lemma:passenger response}
Given an incentive price $p_{s,t}$, a selfish and rational passenger would contribute $q_{i,s}(p_{s,t})=\left[C_{i,s}^{\prime^{-1}}(q_{i,s}(p_{s,t}))\right]^+$ amount of traffic offload to maximize its utility $U_{i,t}(\mathbf{p}_t)$. 
\end{lemma}
\begin{proof}
We omit the proof due to space constraint. See \cite{report} for detailed proof.
\end{proof}
We have the following two observations by Lemma \ref{lemma:passenger response}. First, a selfish and rational passenger that optimizes its utility will contribute the amount of traffic offload $C_{i,s}^{\prime^{-1}}(q_{i,s}(p_{s,t}))$ if and only if it can obtain non-negative utility. Moreover, by observing the participation of each passenger, the government can infer the gradients of inconvenience cost functions.

Taking the amount of traffic offload of each participating passenger $q_{i,s}(p_{s,t})$ as feedback, the government can then use the gradient descent algorithm \cite{zinkevich2003online} to approximately minimize the social cost. In Algorithm \ref{algo:incentive}, the government first initializes a set of learning rates $\{\eta_1,\cdots,\eta_T\}$ that adjusts the step size between two time instants. In the meanwhile, Algorithm \ref{algo:incentive} initializes $\mathbf{p}_1$ of small value for time $t=1$. Then for each time step $t=2,\cdots,T$, the government iteratively updates the incentive price $\mathbf{p}_{t+1}$ as $\max\left\{p_{s,t}-\sum_i\eta_tC_{i,s}^{\prime}(q_{i,s}^*),0\right\}.$

	\begin{algorithm}
		\caption{Computation of incentive price}
		\label{algo:incentive}
		\begin{algorithmic}[1]	
		    \State Initialize the sequence of learning rates $\eta_1,\cdots,\eta_{T-1}$
		    \While {$t\leq T$}
			\State Initialize incentive price $\mathbf{p}_1>0$ for time step $t=1$ arbitrarily
			\State Update incentive price as $p_{s,t+1}=\max\left\{p_{s,t}-\sum_i\eta_tC_{i,s}^{\prime}(q_{i,s}^*),0\right\}$ for all $s$
			\EndWhile
		\end{algorithmic}
	\end{algorithm}

In the following, we characterize Algorithm \ref{algo:incentive} by analyzing the social cost incurred using the incentive price returned by Algorithm \ref{algo:incentive}. Analogous to online convex algorithm \cite{zinkevich2003online}, we define the regret of the government. The regret over time horizon $T$ is defined as
\begin{equation}\label{eq:regret}
    R(T)=\Lambda(\mathbf{p})-\Lambda^*,
\end{equation}
where $\Lambda(\mathbf{p})$ is the social cost when when selecting a sequence of incentive prices $\{p_{s,t}\}_{s=1,t=1}^{S,T}$ as defined in \eqref{eq:gov utility}, and 
\begin{multline}
    \Lambda^*=\min_{p}\sum_t\sum_{s\in\mathcal{S}}\Bigg\{\sum_{i\in\mathcal{N}}C_{i,s}(q_{i,s}(p))\\+\beta_s\left[Q_{s,t}-\sum_{i\in\mathcal{N}}q_{i,s}(p) \right]^+\Bigg\}
\end{multline}
is the optimal social cost when using a fixed price.
Then the regret \eqref{eq:regret} models the difference between the social cost when selecting a sequence of incentive prices $\{\mathbf{p}_{t}\}_{t=1}^{T}$ and optimal social cost from using a fixed price $p_s^*$ for each $s$. 

In the following, we characterize the mechanism design proposed for one-way communication by analyzing the regret \eqref{eq:regret}. In particular, we analyze the regret \eqref{eq:regret} by showing that it satisfies Hannan consistency, i.e., 
\begin{equation}\label{eq:avg regret}
    \limsup_{T\rightarrow\infty}\frac{R(T)}{T}\rightarrow 0.
\end{equation}
Hannan consistency implies that the average regret \eqref{eq:avg regret} vanishes when the time horizon approaches infinity. We define the following notations. Define row vectors $\mathbf{g}_{s,t}\in\mathbb{R}^N$ and $\mathbf{h}_{s,t}\in\mathbb{R}^N$ as :
\begin{align}
    \mathbf{g}_{s,t}&=\left[C'_{1,s}(q_{i,s}(p_{s,t})),\cdots,C'_{N,s}(q_{i,s}(p_{s,t}))\right]\label{eq:g definition}\\
    \mathbf{h}_{s,t}&=\left[q'_{1,s}(p_{s,t}),\cdots,q'_{N,s}(p_{s,t})\right].\label{eq:h definition}
\end{align}
We denote the vectors $\mathbf{g}_{s,t}$ and $\mathbf{h}_{s,t}$ that are associated with $p_{s,t}=p_s^*$ as $\mathbf{g}^*_{s,t}$ and $\mathbf{h}^*_{s,t}$, respectively. Let $\Bar{g}=\max_{s,t}\mathbf{g}_{s,t}(p_{s,t})$ and $\underline{g}=\min_{s,t}g_{s,t}$. Denote the maximum incentive price the government would issue as $\Bar{p}$. We also define column vectors for all $s$ and $t$ as $\mathbf{q}_{s,t}=\left[q_{1,s}(p_{s,t}),\cdots,q_{N,s}(p_{s,t})\right]^T$.
Similarly, vector $\mathbf{q}^*_{s,t}$ represents the vector associated with $p_{s,t}=p_s^*$. We finally define $k_{s,t}=\mathbf{g}_{s,t}\cdot \mathbf{h}_{s,t}+\beta_s\mathbf{h}_{s,t}\mathbf{1}_N$, where $\mathbf{g}_{s,t}\cdot \mathbf{h}_{s,t}$ is the dot product of $\mathbf{g}_{s,t}$ and $\mathbf{h}_{s,t}$. Let $\Bar{k}=\max_{s,t}k_{s,t}$ be the maximum $k_{s,t}$ for all $s$ and $t$. Next we show that regret \eqref{eq:regret} is upper bounded.

\begin{lemma}\label{lemma:regret}
The regret of Algorithm \ref{algo:incentive} is bounded as
\begin{equation}\label{eq:regret bound}
    R(T)\leq \sum_s\left\{\frac{\Bar{p}^2k_{s,T}}{2\eta_T\mathbf{g}_{s,T}\mathbf{1}_N}+\sum_{t=1}^T\frac{\eta_t\Bar{g}^2N^2k_{s,t}}{2\mathbf{g}_{s,t}\mathbf{1}_N}\right\}.
\end{equation}
\end{lemma}
\begin{proof}
The proof is motivated by \cite{zinkevich2003online}. Denote the optimal incentive price associated with optimal social cost $\Lambda^*$ as $p_{s}^*$ for each OD pair $s$. Due to convexity of inconvenience cost functions $C_{i,s}(\cdot)$, for any $q_{i,s}(p_{s,t})$ and $p_{s,t}$ we have 
\begin{multline*}
    \sum_iC_{i,s}(q_{i,s}(p_{s,t}))+\beta_s\left[Q_{s,t}-\sum_{i\in\mathcal{N}}q_{i,s}(p_{s,t}) \right]^+\\
    \geq\sum_i\bigg\{C'_{i,s}\left(q^*_{i,s}\right)\left(q_{i,s}(p_{s,t})-q^*_{i,s}\right)+C_{i,s}(q^*_{i,s})\bigg\}\\
    +\beta_s\left[Q_{s,t}-\sum_{i\in\mathcal{N}}q_{i,s}(p_{s,t}) \right]^+.
\end{multline*}
By definition of $\mathbf{g}_{s,t}$ \eqref{eq:g definition}, we have that the optimal social cost satisfies the following inequalities:
\begin{align}
    &\sum_iC_{i,s}(q_{i,s}(p^*_{s}))+\beta_s\left[Q_{s,t}-\sum_{i\in\mathcal{N}}q_{i,s}(p^*_{s})\right]^+\nonumber\\
    \geq&~\mathbf{g}_{s,t}\left(\mathbf{q}^*_{s,t}-\mathbf{q}_{s,t}\right)+\sum_iC_{i,s}\left(q_{i,s}(p_{s,t})\right)\nonumber\\
    &+\beta_s\left[Q_{s,t}-\sum_{i\in\mathcal{N}}q_{i,s}(p^*_{s}) \right]^+\label{eq:Lemma3 ineq 1}\\
    \geq&~\mathbf{g}_{s,t}\left(\mathbf{q}^*_{s,t}-\mathbf{q}_{s,t}\right)+\sum_iC_{i,s}(q_{i,s}(p_{s,t}))\nonumber\\
    &+\beta_s\left[Q_{s,t}-\mathbf{h}_{s,t}\mathbf{1}_N(p^*_{s}-p_{s,t})-\sum_iq_{i,s}(p_{s,t})\right]^+\label{eq:Lemma3 ineq 2}\\
    \geq&~\mathbf{g}_{s,t}\left(\mathbf{q}^*_{s,t}-\mathbf{q}_{s,t}\right)+\sum_iC_{i,s}(q_{i,s}(p_{s,t}))\nonumber\\
    &+\beta_s\left[Q_{s,t}-\sum_iq_{i,s}(p_{s,t})\right]^+-\beta_s\left[\mathbf{h}_{s,t}\mathbf{1}_N(p^*_{s}-p_{s,t})\right]^+,\label{eq:Lemma3 ineq 3}
\end{align}
where $\mathbf{1}_N=[1,\cdots,1]^T$ with dimension $N$, inequality \eqref{eq:Lemma3 ineq 1} follows by the convexity of $C_{i,s}(\cdot)$, inequality \eqref{eq:Lemma3 ineq 2} follows by the first order Taylor expansion of concave function $q_{i,s}(\cdot)$, and inequality \eqref{eq:Lemma3 ineq 3} holds by the fact that $[a-b]^+\geq [a]^+-[b]^+$. Rearranging the inequality above, we have that
\begin{align*}
    &\sum_iC_{i,s}(q_{i,s}(p_{s,t}))+\beta_s\left[Q_{s,t}-\sum_{i\in\mathcal{N}}q_{i,s}(p_{s,t})\right]^+\\
    &-\sum_iC_{i,s}(q_{i,s}(p^*_{s}))-\beta_s\left[Q_{s,t}-\sum_{i\in\mathcal{N}}q_{i,s}(p^*_{s}) \right]^+\\
    \leq~& \beta_s\left[\mathbf{h}_{s,t}\mathbf{1}_N(p^*_{s}-p_{s,t})\right]^+-\mathbf{g}_{s,t}\left(\mathbf{q}^*_{s,t}-\mathbf{q}_{s,t}\right)\\
    \leq~& \beta_s\left[\mathbf{h}_{s,t}\mathbf{1}_N(p^*_s-p_{s,t})\right]^+-\mathbf{g}_{s,t}\cdot \mathbf{h}_{s,t}(p^*_s-p_{s,t})\\
    \leq~&k_{s,t}(p^*_s-p_{s,t}),
\end{align*}
where $k_{s,t}=I_{\{p^*\geq p_{s,t}\}}\beta_s\mathbf{h}_{s,t}\mathbf{1}_N-\mathbf{g}_{s,t}\cdot \mathbf{h}_{s,t}$, $I_{\{p^*\geq p_{s,t}\}}$ is an indicator that equals to $1$ if $p^*\geq p_{s,t}$ and $0$ otherwise, and $\mathbf{g}_{s,t}\cdot \mathbf{h}_{s,t}$ represents the dot product of $\mathbf{g}_{s,t}$ and $\mathbf{h}_{s,t}$. At time step $t+1$, we have
\begin{align}
    &(p_{s,t+1}-p^*_{s})^2
    \leq\left(p_{s,t}-\eta_t\mathbf{g}_{s,t}\mathbf{1}_N-p^*_s\right)^2\label{eq:Lemma3 ineq 4}\\
    \leq~&(p_{s,t}-p^*_s)^2-2\eta_t\mathbf{g}_{s,t}\mathbf{1}_N(p_{s,t}-p^*_s)+\eta_t^2\Bar{g}^2N^2,\label{eq:Lemma3 ineq 5}
\end{align}
where inequality \eqref{eq:Lemma3 ineq 4} holds by the updating rule of $p_{s,t}$ and inequality \eqref{eq:Lemma3 ineq 5} holds due to $\mathbf{g}_{s,t}\mathbf{1}_N\leq \Bar{g}N$. Then we obtain
$
    \mathbf{g}_{s,t}\mathbf{1}_N(p^*_s-p_{s,t})\leq \frac{1}{2\eta_t}\Big[(p_{s,t}-p^*_s)^2-(p_{s,t+1}-p^*_{s})^2+\eta_t^2\Bar{g}^2N^2\Big].
$
By \eqref{eq:regret}, we have
\begin{align*}
    &R(T)=\Lambda(\mathbf{p})-\Lambda^*\\
    =~&\sum_t\sum_s\Bigg\{\sum_iC_{i,s}(q_{i,s}(p_{s,t}))+\beta_s\left[Q_{s,t}-\sum_{i\in\mathcal{N}}q_{i,s}(p_{s,t}) \right]^+\\
    &\quad-\sum_iC_{i,s}(q_{i,s}(p^*_{s}))-\beta_s\left[Q_{s,t}-\sum_{i\in\mathcal{N}}q_{i,s}(p^*_{s}) \right]^+\Bigg\}\\
    \leq~&\sum_s\frac{\Bar{p}^2}{2}\bigg\{\frac{k_{s,1}}{\eta_1\mathbf{g}_{s,1}\mathbf{1}_N}+\sum_{t=2}^T\left(\frac{k_{s,t}}{\eta_t\mathbf{g}_{s,t}\mathbf{1}_N}-\frac{k_{s,t-1}}{\eta_{t-1}\mathbf{g}_{s,t-1}\mathbf{1}_S}\right)\\
    &\quad\quad\quad\quad\quad\quad\quad\quad\quad\quad\quad\quad+\sum_{t=1}^T\frac{\eta_t\Bar{g}^2N^2k_{s,t}}{2\mathbf{g}_{s,t}\mathbf{1}_N}\bigg\}\\
    =~&\sum_s\left\{\frac{\Bar{p}^2k_{s,T}}{2\eta_T\mathbf{g}_{s,T}\mathbf{1}_N}+\sum_{t=1}^T\frac{\eta_t\Bar{g}^2N^2k_{s,t}}{2\mathbf{g}_{s,t}\mathbf{1}_N}\right\},
\end{align*}
which completes our proof.
\end{proof}

Leveraging Lemma \ref{lemma:regret}, we are ready to show Hannan consistency holds for the proposed incentive mechanism design.
\begin{proposition}
Let $\eta_t=\frac{1}{\sqrt{t}}$. The regret defined in \eqref{eq:regret} along with the incentive design proposed in Algorithm \ref{algo:incentive} achieves the Hannan consistency.
\end{proposition}
\begin{proof}
We omit the proof due to space constraints. See \cite{report} for detailed proof.
\end{proof}

\subsection{Incentive Mechanism Design with Privacy Guarantees}

In this subsection, we give the differentially private incentive price $p_{s,t}$ under the one-way communication setting.


To achieve the privacy guarantee, we perturb the incentive price returned by Algorithm \ref{algo:incentive} as follows:
\begin{equation}\label{eq:privacy price}
    p_{s,t} = p^*_{s,t}+\delta_t,~\forall s,t
\end{equation}
where $\delta_t\sim\mathcal{L}\left(\frac{\Delta p}{\epsilon}\right)$ is a random variable that follows Laplace distribution with scale $\Delta p/\epsilon$, and $\Delta p$ is the maximum difference of the incentive price under two set of observations that differ in one passenger, which can be obtained by solving
\begin{align*}
    \max_{s,t,\mathbf{q}_{s,t},\mathbf{q}'_{s,t}} &p_{s,t+1}-p'_{s,t+1}\\
    \st \quad&\|\mathbf{q}_{s,t}-\mathbf{q}'_{s,t}\|_1=1
\end{align*}
where $p_{s,t+1}$ and $p'_{s,t+1}$ are incentive prices returned by Algorithm \ref{algo:incentive} given traffic offloads $\mathbf{q}_{s,t}$ and $\mathbf{q}'_{s,t}$, respectively.

In the sequel, we show differential privacy is preserved.
\begin{theorem}
Incentive price design \eqref{eq:privacy price} achieves $\left(\left(T-\sum_{t=1}^{T-1}\eta_t\right)\epsilon\right)$-differential privacy.
\end{theorem}
\begin{proof}
We prove by induction. We first prove that differential privacy holds for single time step. Then we generalize the analysis on one time step scenario to multiple time steps scenario. Since the passengers' utility function is deterministic, given an incentive price $p_{s,t}$, passengers' participation is deterministic. Given the initial incentive price $p_{s,1}$ at $t=1$, the contribution of passenger $i$ is determined as $q_{i,s}(p_{s,1})$. 
We compare the p.d.f.'s at $p_{s,2}=p'_{s,2}$.
\begin{align*}
    &\frac{P\left(p_{s,2}\right)}{P'\left(p'_{s,2}\right)}\\
    =~&\frac{\exp\left(-\frac{\epsilon|p_{s,1}-\sum_j\eta_1C_{j,s}^{\prime}(q_{j,s}(p_{s,1}))-p_{s,2}|}{\Delta p}\right)}{\exp\left(-\frac{\epsilon|p'_{s,1}-\sum_j\eta_1C_{j,s}^{\prime}(q_{j,s}(p'_{s,1}))-p_{s,2}|}{\Delta p}\right)}\\
    \leq~& \exp\Big(-\epsilon|p_{s,1}-\sum_j\eta_1C_{i,s}^{\prime}(q_{i,s}(p_{s,1}))+\\
    &\quad\left(p'_{s,1}-\sum_j\eta_1C_{j,s}^{\prime}(q_{j,s}(p'_{s,1}))\right)|/\Delta p\Big)\\
    =~&\exp\bigg(\epsilon\Big\{|p'_{s,1}-p_{s,1}+\eta_1\big[C'_{i,s}(q_{i,s}(p'_{s,1}))\\
    &\quad\quad\quad\quad-C'_{i,s}(q_{i,s}(p_{s,1}))\big]|\Big\}/\Delta p\bigg)\\
    =~&\exp((1-\eta_1)\epsilon),
\end{align*}
where the inequality follows from triangle inequality, and the last equality follows by Lemma \ref{lemma:passenger response}. Thus we have $(1-\eta_1)\epsilon$-differential privacy.

We note that since the scheme follows Stackelberg setting, a malicious party can only infer the passengers' behavior at time $t=1$ by observing $p_{s,2}$. Thus, the analysis on single time step serves as our induction base.

At time $t$, the information perceived by the malicious party is $\mathcal{I}_{t}^{mal}=\{p_{s,t'}|\forall s,t=1,\cdots,t\}$. We analyze the ratio of $\frac{P\left(p_{s,t}\right)}{P'\left(p'_{s,t}\right)}$ under the under the following scenarios. First, if $p_{s,t'}=p'_{s,t'}$ for all $t'<t$ and $p_{s,t}$ distinguishes from $p_{s,t}$, then we have $(1-\eta_t)$-differential privacy. In the following, we focus on the general setting in which $p_{s,t'}$ differs from $p'_{s,t'}$ for all $t'<t$ such that $\mathbf{q}_{s,1:t}$ and $\mathbf{q}'_{s,1:t}$ differ at at most one entry. Then we have

\begin{align*}
    &\frac{P\left(p_{s,t}\right)}{P'\left(p'_{s,t}\right)}
    =\prod_{\tau=2}^t\left(\frac{Pr\left(p_{s,\tau}|p_{s,\tau-1}\right)}{Pr\left(p'_{s,\tau}|p'_{s,\tau-1}\right)}\right)\\
    =~&\prod_{\tau=1}^t\bigg\{\exp\Big[\epsilon\big[|p_{s,\tau}-p'_{s,\tau}+\eta_\tau\big(C'_{i,s}(q_{i,s}(p'_{s,\tau}))\\
    &\quad\quad\quad-C'_{i,s}(q_{i,s}(p_{s,\tau})\big)\big]|/\Delta p\Big]\bigg\}\\
    =~&\exp\left(\left(T-\sum_{\tau=1}^t\eta_\tau\right)\epsilon\right).
\end{align*}

Therefore, we have that the proposed approach achieves $\left((T-\sum_{t=1}^{T-1}\eta_t)\epsilon\right)$-differential privacy.
\end{proof}

In the remainder of this section, we characterize the social welfare using the incentive design \eqref{eq:privacy price}. We start with the expected regret defined as the probabilistic counter-part of \eqref{eq:regret}:
$
    \mathbb{E}\{R(T)\}=\mathbb{E}_{\mathbf{p}}\{\Lambda(\mathbf{p})\}-\Lambda^*,
$
where $\mathbb{E}_{\mathbf{p}}\{\cdot\}$ represents expectation with respect to $\mathbf{p}$.
\begin{lemma}
The expected regret $\mathbb{E}\{R(T)\}$ under incentive \eqref{eq:privacy price} is bounded from above as
\begin{equation}\label{eq:expected regret bound}
    \mathbb{E}\{R(T)\}\leq \mathbb{E}\left\{\sum_s\left[\frac{\Bar{p}^2k_{s,T}}{2\eta_T\mathbf{g}_{s,T}\mathbf{1}_N}+\sum_{t=1}^T\frac{\eta_t\Bar{g}^2N^2k_{s,t}}{2\mathbf{g}_{s,t}\mathbf{1}_N}\right]\right\}.
\end{equation}
\end{lemma}
\begin{proof}
The proof is the probabilistic counter-part of that of Lemma \ref{lemma:regret}. We omit the proof due to space constraints.
\end{proof}

Before closing this section, we finally show Hannan consistency holds under incentive design \eqref{eq:privacy price}, i.e., 
\begin{equation}\label{eq:Hannan consistency}
    \limsup_{T\rightarrow\infty}\frac{R(T)}{T}=0\mbox{ with probability one.}
\end{equation}

\begin{theorem}\label{thm:Hannan consistency}
The Hannan consistency \eqref{eq:Hannan consistency} holds for incentive design \eqref{eq:privacy price}.
\end{theorem}

To prove Theorem \ref{thm:Hannan consistency}, we first give the following lemma.
\begin{lemma}\label{lemma:limsup and probability}
Let $Pr(\cdot)$ be the probability of an event. Then the following inequality holds
\begin{multline}\label{eq:limsup and probability}
    Pr\left(\limsup_{T\rightarrow\infty}\left\{\sum_{t=1}^TS\Bar{k}\max_{s}\|p^*_s-p_{s,t}\|_\infty/T\right\}\leq 0\right)\\
    \geq\limsup_{T\rightarrow\infty}Pr\left(\sum_{t=1}^TS\Bar{k}\max_{s}\|p^*_s-p_{s,t}\|_\infty/T\leq 0\right)
\end{multline}
\end{lemma}
\begin{proof}
We omit the proof due to space constraints. See \cite{report} for detailed proof.
\end{proof}

Now we are ready to prove Theorem \ref{thm:Hannan consistency}.
\begin{proof}
\emph{(Proof of Theorem \ref{thm:Hannan consistency}.)} Let $\Bar{k}=\max_{s,t}k_{s,t}$ be the maximum $k_{s,t}$ for all $s$ and $t$, and $\underline{g}=\min_{s,t}g_{s,t}$. Then following the proof of Lemma \ref{lemma:regret}, we have
\begin{align*}
    &\sum_iC_{i,s}(q_{i,s}(p_{s,t}))+\beta_s\left[Q_{s,t}-\sum_{i\in\mathcal{N}}q_{i,s}(p_{s,t})\right]^+\\
    &-\sum_iC_{i,s}(q_{i,s}(p^*_{s}))-\beta_s\left[Q_{s,t}-\sum_{i\in\mathcal{N}}q_{i,s}(p^*_{s}) \right]^+\\
    \leq~&k_{s,t}(p^*_s-p_{s,t})
    \leq\Bar{k}(p^*_s-p_{s,t}).
\end{align*}
Summing the inequality above over $t$ and $s$, we have
\begin{align*}
    &\Lambda(\mathbf{p})-\Lambda^*\\
    =~&\sum_t\sum_s\Bigg(\sum_iC_{i,s}(q_{i,s}(p_{s,t}))+\beta_s\left[Q_{s,t}-\sum_{i\in\mathcal{N}}q_{i,s}(p_{s,t})\right]^+\\
    &-\sum_iC_{i,s}(q_{i,s}(p^*_{s}))-\beta_s\left[Q_{s,t}-\sum_{i\in\mathcal{N}}q_{i,s}(p^*_{s}) \right]^+\Bigg)\\
    \leq~&\sum_t\sum_s\left(k_{s,t}(p^*_s-p_{s,t})\right)\\
    \leq~&\sum_t\sum_s\left(\Bar{k}(p^*_s-p_{s,t})\right)
    \leq\sum_t\sum_s\left(\Bar{k}\|p^*_s-p_{s,t}\|_\infty\right).
\end{align*}

Let $Pr(\cdot)$ be the probability of an event. Then we have

\begin{align}
    &Pr\left(\limsup_{T\rightarrow\infty}\frac{R(T)}{T}\leq0\right)\nonumber\\
    \geq~&Pr\left(\limsup_{T\rightarrow\infty}\left\{\sum_t\sum_s\left(\Bar{k}\|p^*_s-p_{s,t}\|_\infty\right)/T\right\}\leq0\right)\nonumber\\
    \geq~&Pr\left(\limsup_{T\rightarrow\infty}\left\{\sum_tS\Bar{k}\max_{s}\|p^*_s-p_{s,t}\|_\infty/T\right\}\leq 0\right)\nonumber\\
    \geq~&\limsup_{T\rightarrow\infty}\left\{Pr\left(\sum_t\frac{S\Bar{k}\max_{s}\|p^*_s-p_{s,t}\|_\infty}{T}\leq 0\right)\right\}\label{eq:thm4 ineq 1}\\
    \geq~&\limsup_{T\rightarrow\infty}\left\{1-Pr\left(\sum_t\frac{S\Bar{k}\max_{s}\|p^*_s-p_{s,t}\|_\infty}{T}\geq \frac{\Delta p}{\epsilon}\right)\right\}\label{eq:thm4 ineq 2}\\
    =~&\limsup_{T\rightarrow\infty}\left\{1-\exp\left(-\frac{T}{S\Bar{k}}\right)\right\}
    =1,\nonumber
\end{align}
where inequality \eqref{eq:thm4 ineq 1} holds by Lemma \ref{lemma:limsup and probability}, inequality \eqref{eq:thm4 ineq 2} holds by \eqref{eq:privacy price} and definition of $\delta_t$. Therefore, we have Hannan consistency holds.
\end{proof}

\section{Numerical Case Study}\label{sec:simulation}
\begin{figure*}[t!]
\centering
                 \begin{subfigure}{.65\columnwidth} \includegraphics[width=\columnwidth]{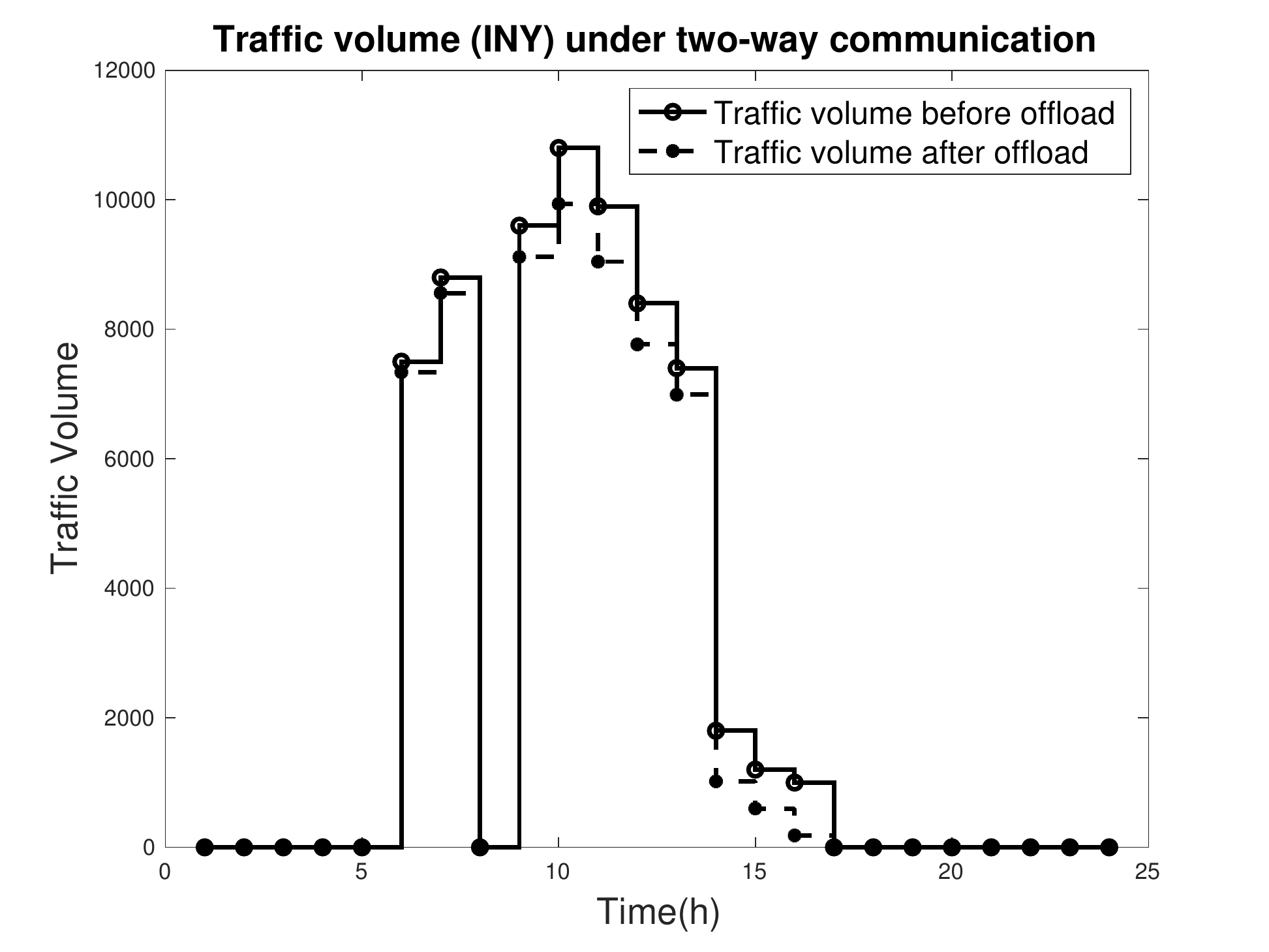}
                 \subcaption {}
                 \label{fig:two way iny}
                 \end{subfigure}\hfill
                 \begin{subfigure}{.65\columnwidth} \includegraphics[width=\columnwidth]{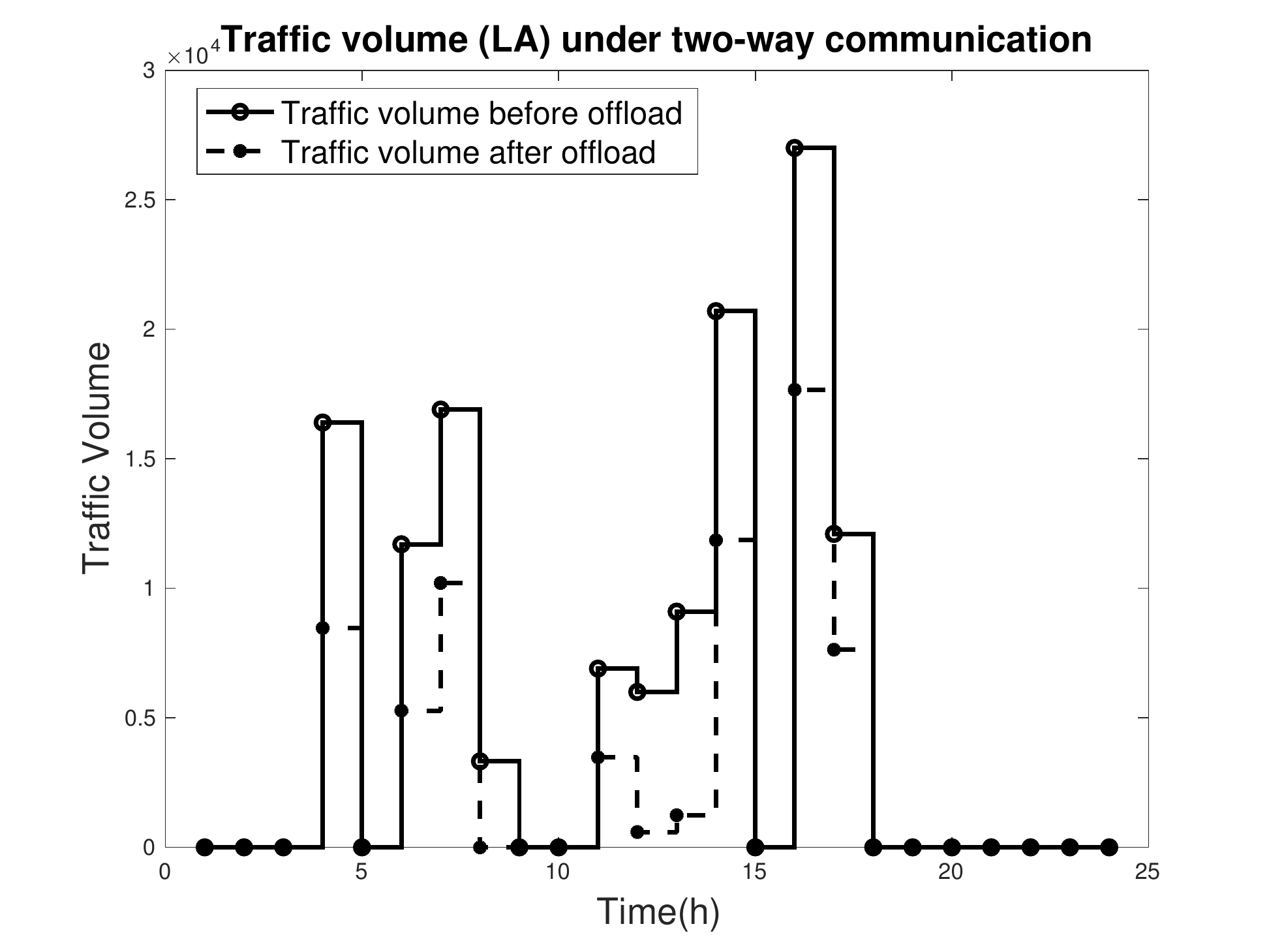}
                 \subcaption {}
                 \label{fig:two way la}
                 \end{subfigure}\hfill
                 \begin{subfigure}{.65\columnwidth} \includegraphics[width=\columnwidth]{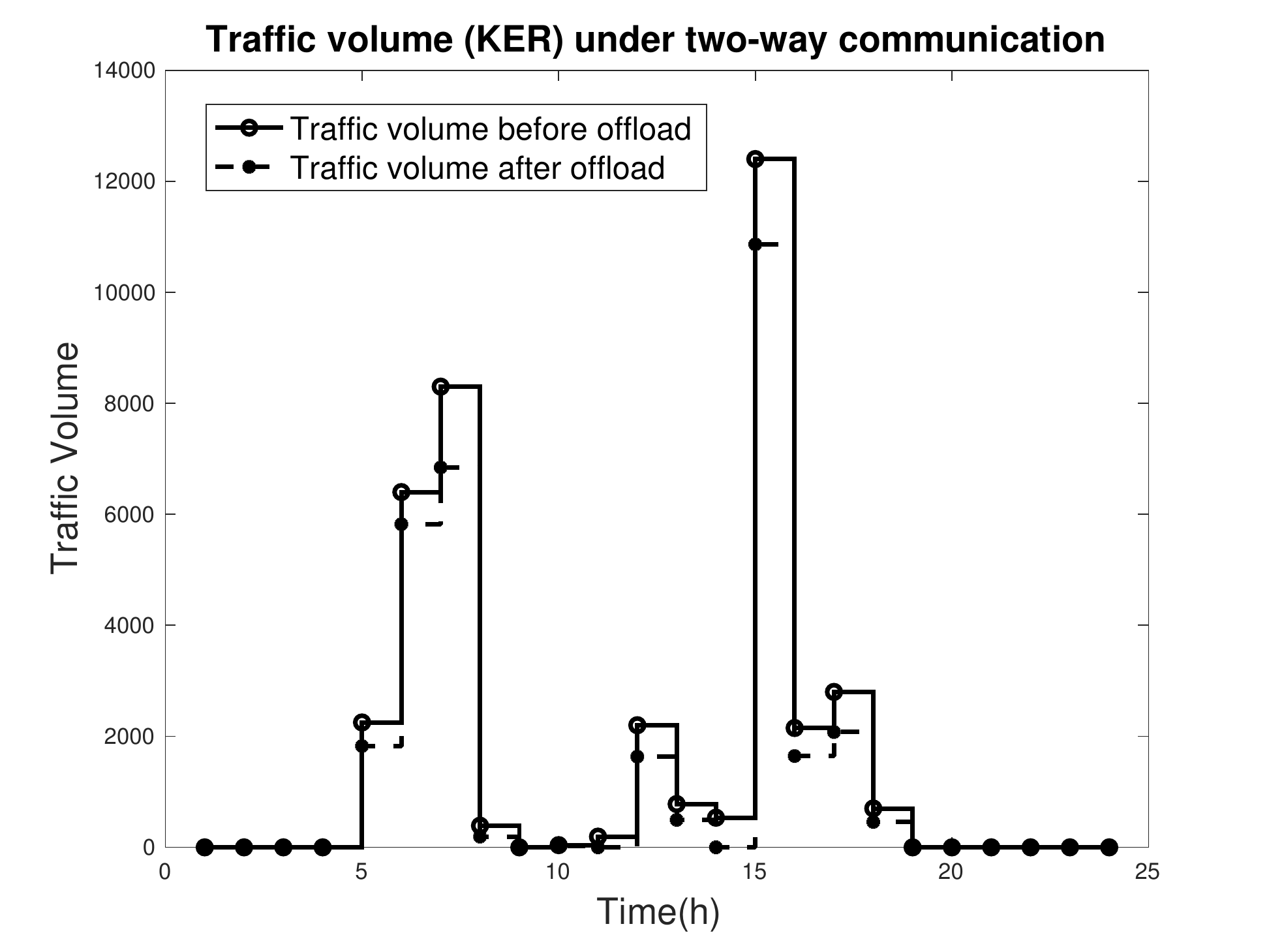}
                 \subcaption{}
                 \label{fig:two way ker}
                 \end{subfigure}\hfill
                 \begin{subfigure}{.65\columnwidth} \includegraphics[width=\columnwidth]{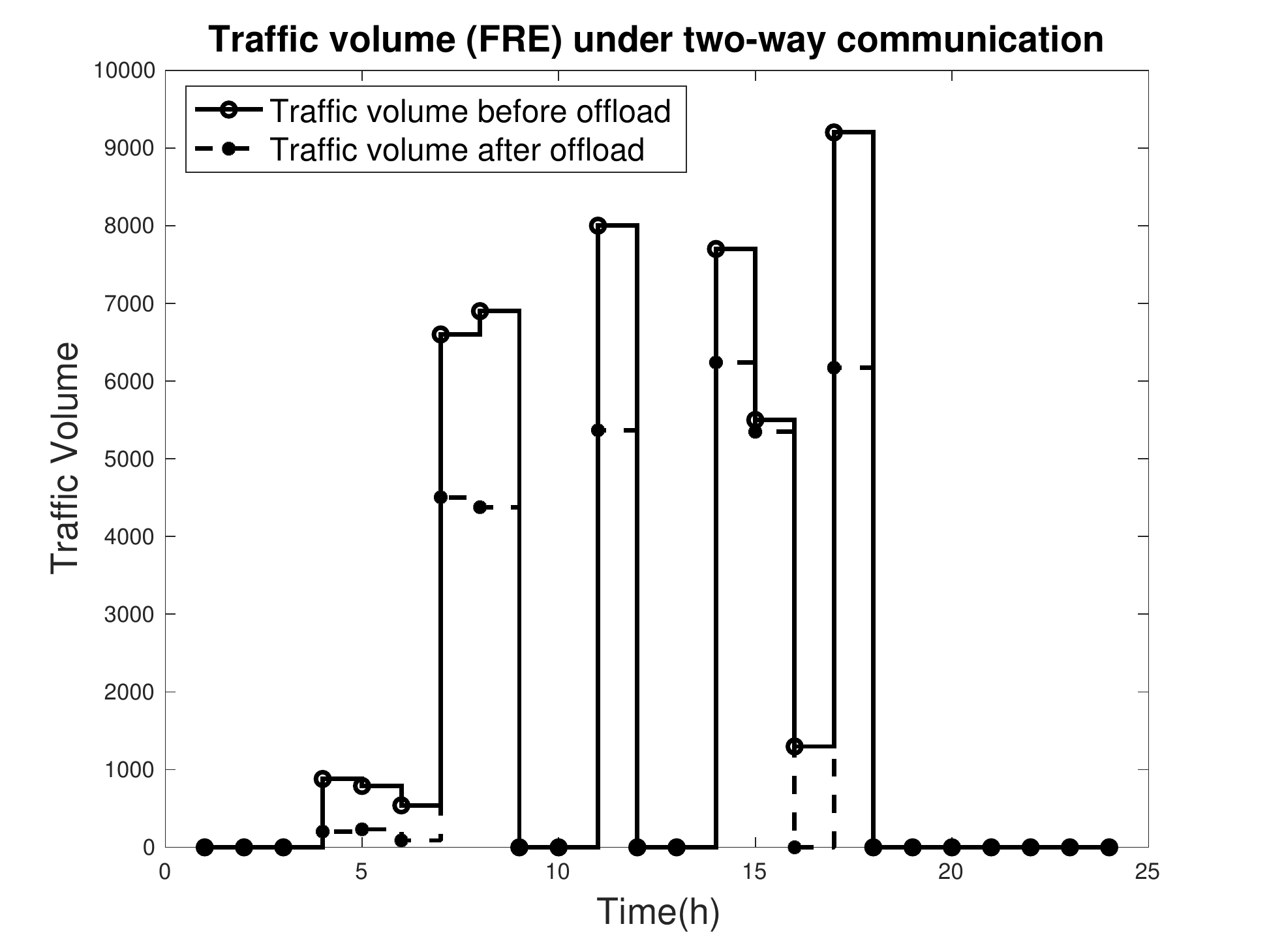}
                 \subcaption{}
                 \label{fig:two way fre}
                 \end{subfigure}\hfill
                 \begin{subfigure}{.65\columnwidth} \includegraphics[width=\columnwidth]{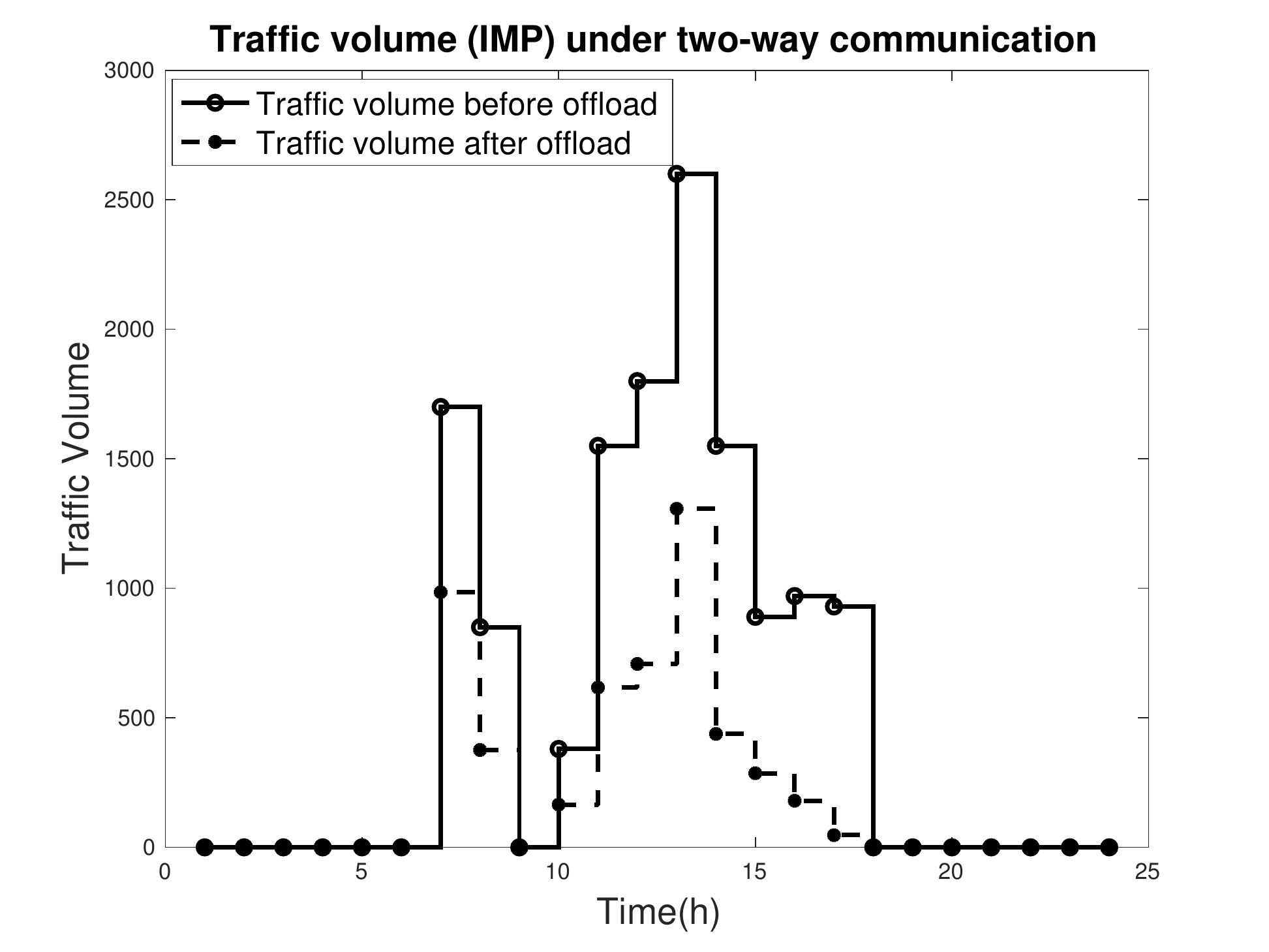}
                 \subcaption{}
                 \label{fig:two way imp}
                 \end{subfigure}\hfill
\caption{In Fig. \ref{fig:two way iny} to Fig. \ref{fig:two way imp}, we present the traffic volume before and after traffic offload. The solid curve is the traffic volume before traffic offload, whereas the dashed curve represents the traffic volume after traffic offload.}
\end{figure*}

\begin{figure*}[t!]
\centering
                 \begin{subfigure}{.65\columnwidth}
                 \includegraphics[width=\columnwidth]{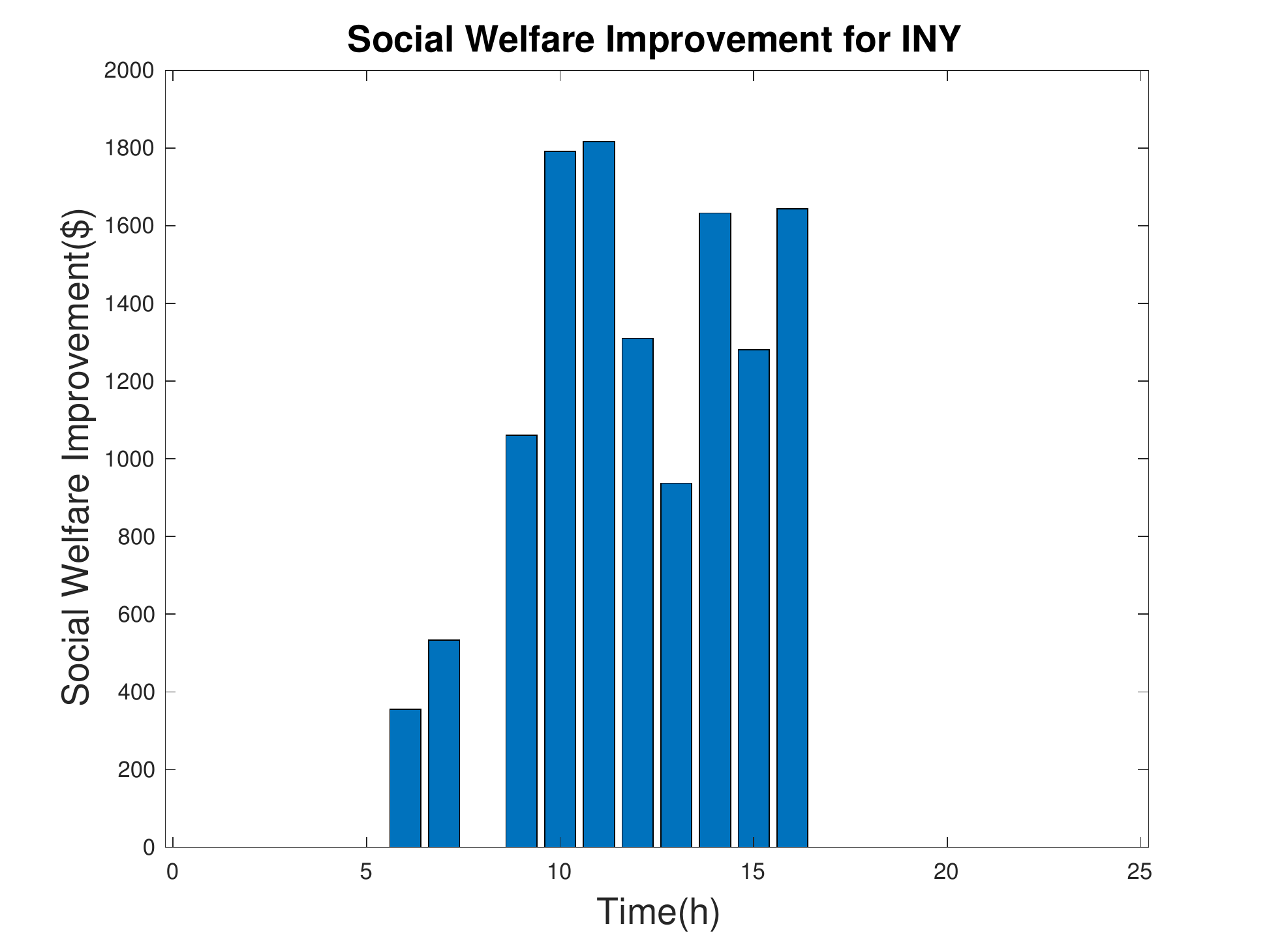}
                 \subcaption {}
                 \label{fig:two way iny sw}
                 \end{subfigure}\hfill
                 \begin{subfigure}{.65\columnwidth}
                 \includegraphics[width=\columnwidth]{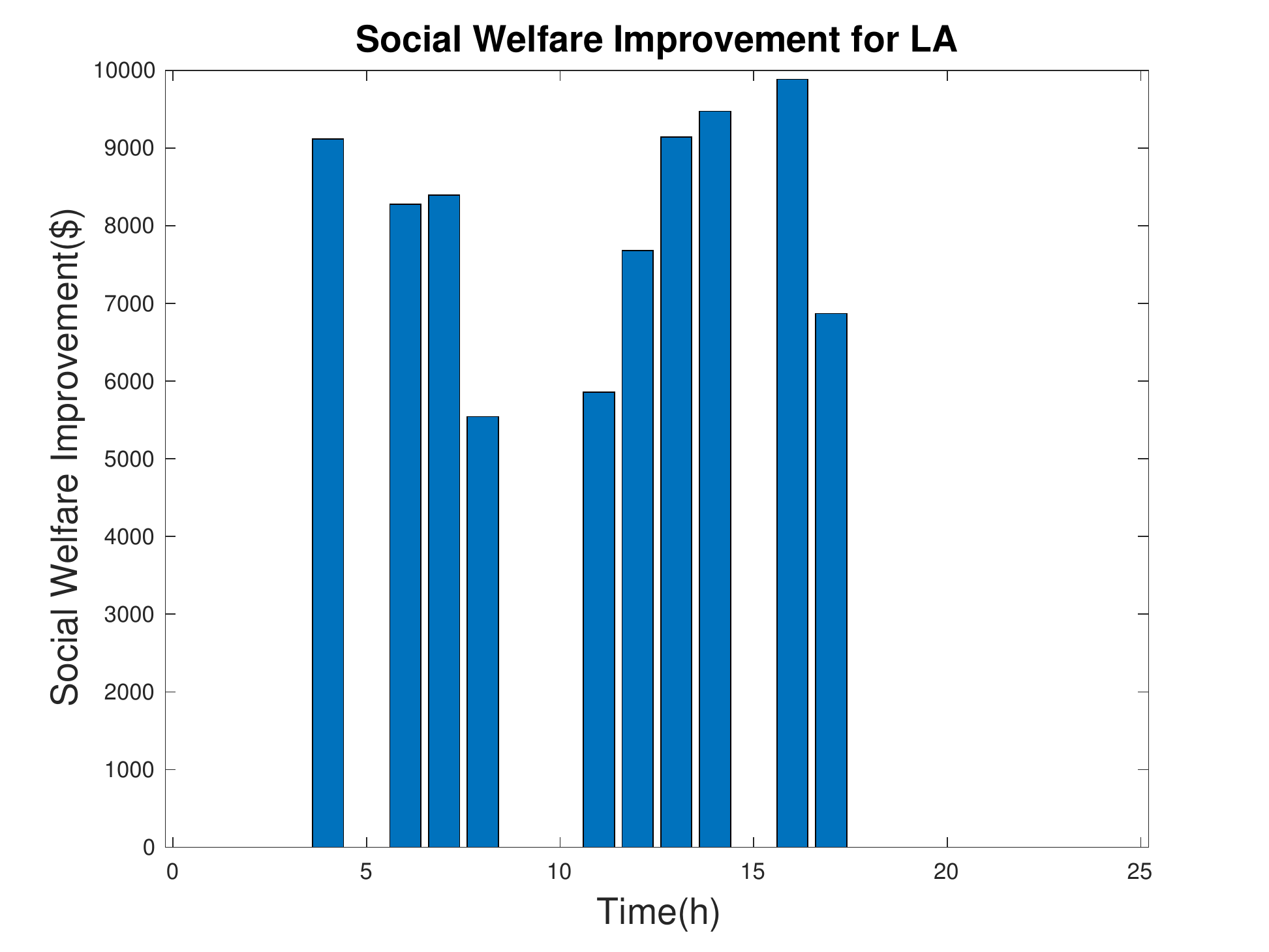}
                 \subcaption {}
                 \label{fig:two way la sw}
                 \end{subfigure}\hfill
                 \begin{subfigure}{.65\columnwidth}
                 \includegraphics[width=\columnwidth]{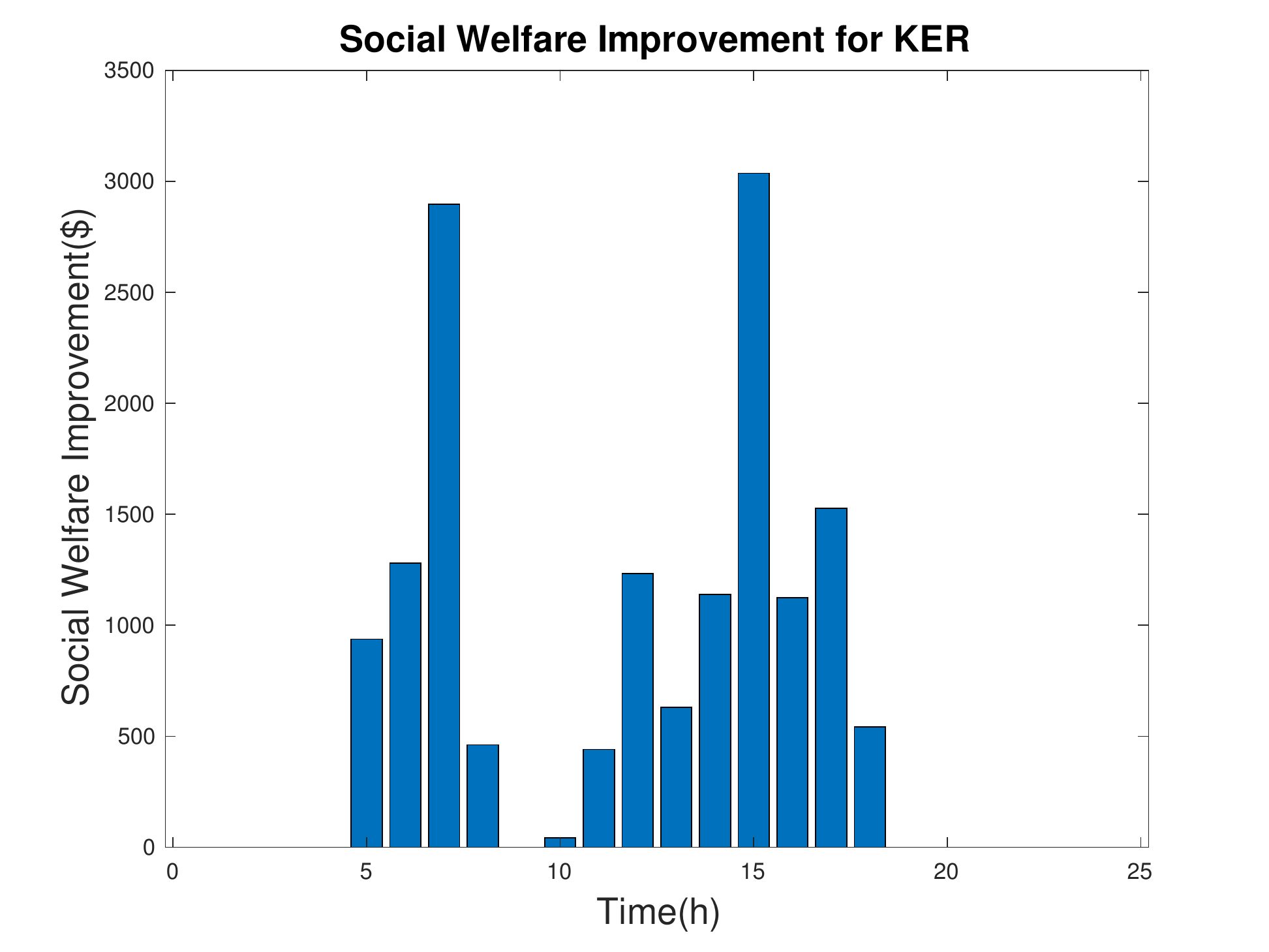}
                 \subcaption{}
                 \label{fig:two way ker sw}
                 \end{subfigure}\hfill
                 \begin{subfigure}{.65\columnwidth} \includegraphics[width=\columnwidth]{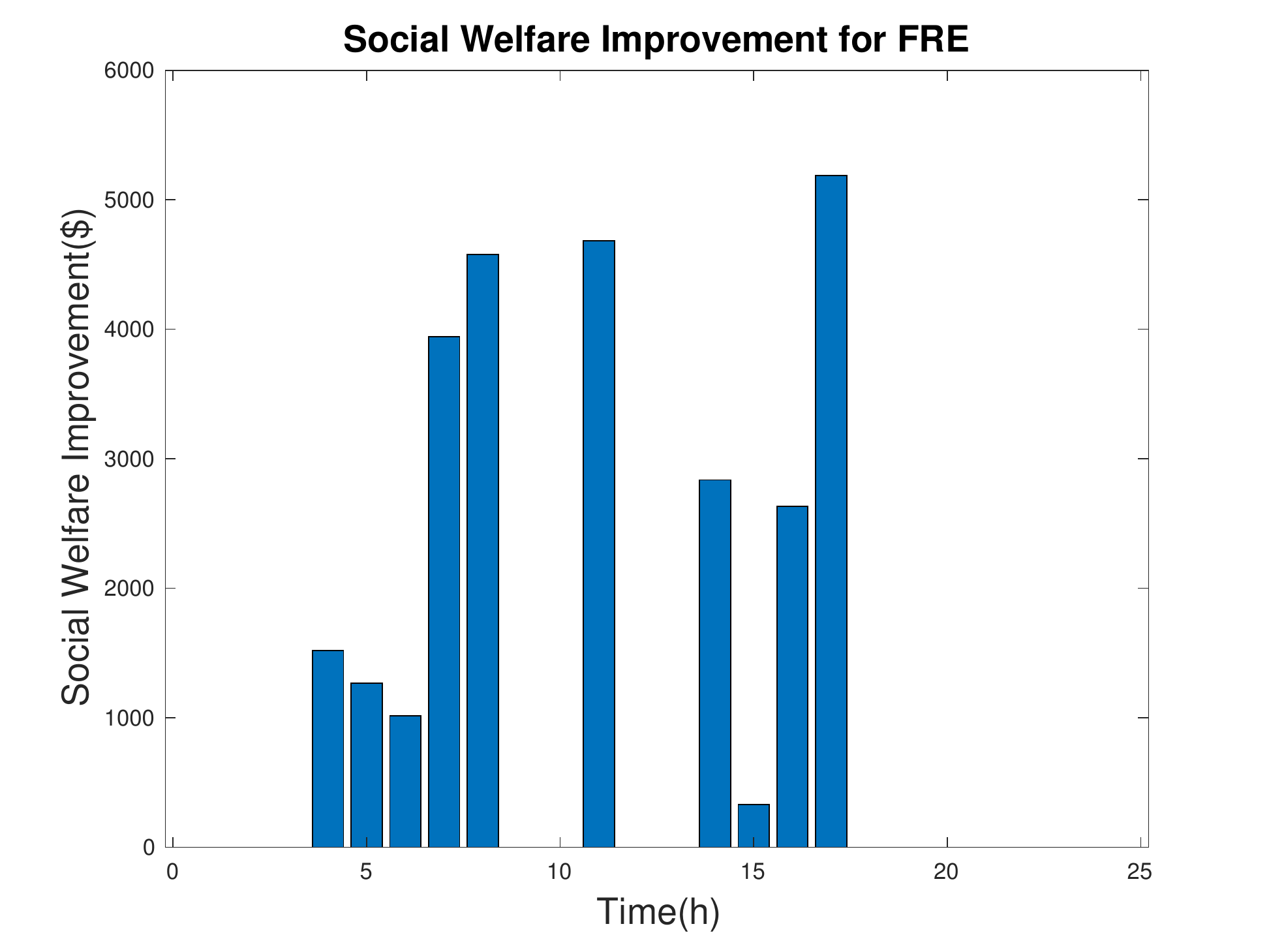}
                 \subcaption{}
                 \label{fig:two way fre sw}
                 \end{subfigure}\hfill
                 \begin{subfigure}{.65\columnwidth} \includegraphics[width=\columnwidth]{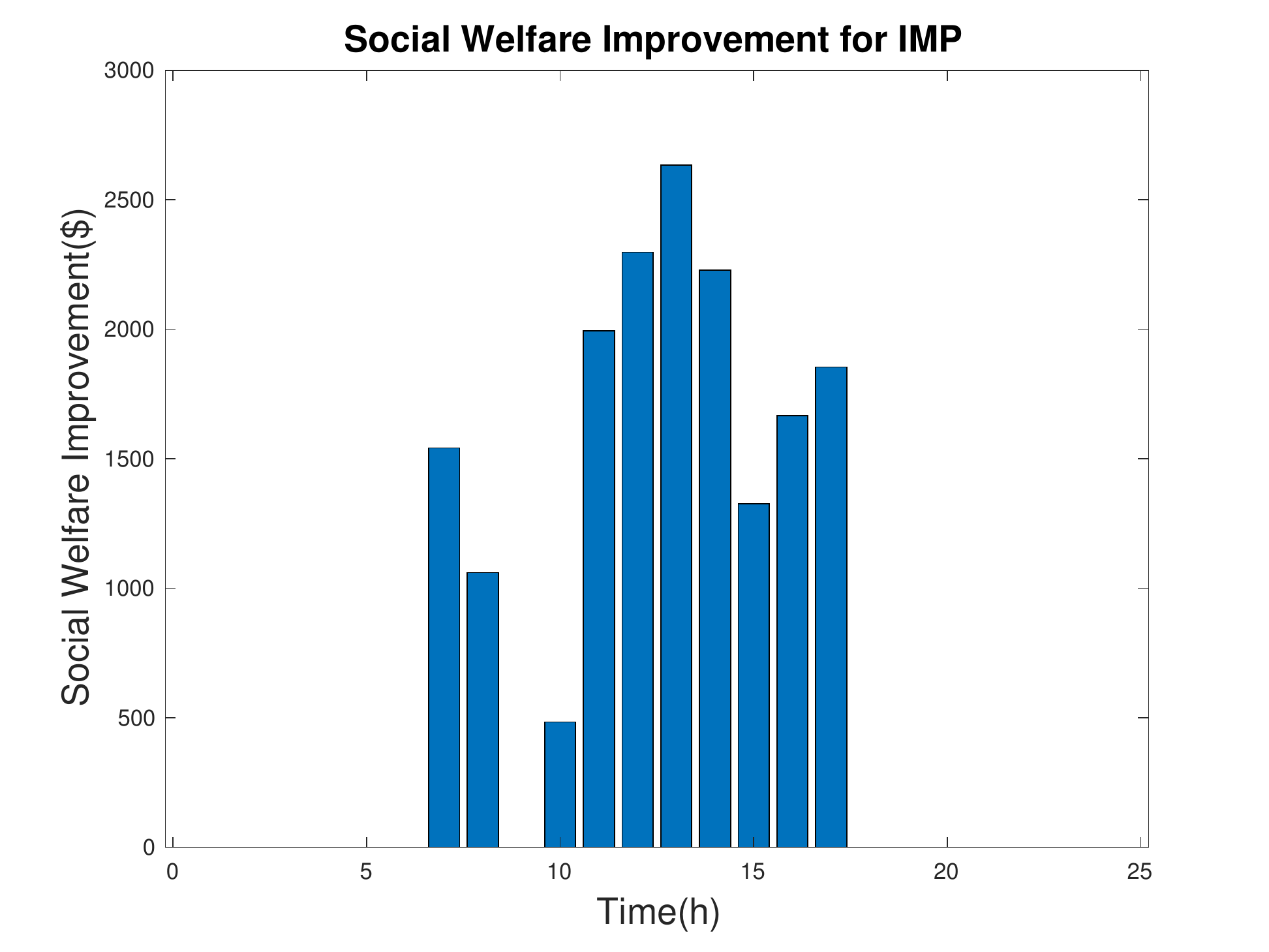}
                 \subcaption{}
                 \label{fig:two way imp sw}
                 \end{subfigure}\hfill
\caption{In Fig. \ref{fig:two way iny sw} to Fig. \ref{fig:two way imp sw}, we present the social welfare at each OD pair.}
\end{figure*}

\begin{figure*}[t!]
\centering
                 \begin{subfigure}{.65\columnwidth} \includegraphics[width=\columnwidth]{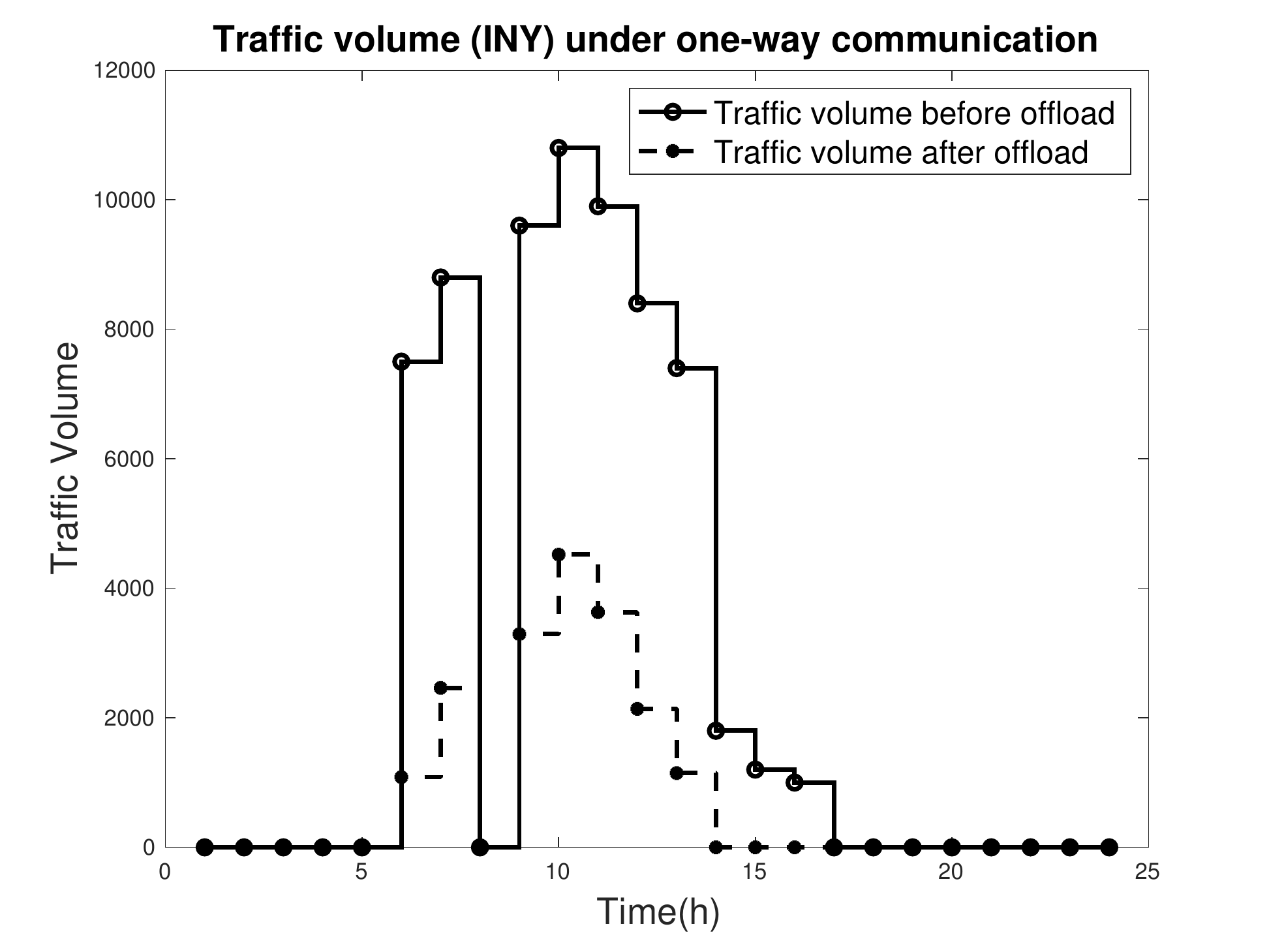}
                 \subcaption {}
                 \label{fig:one way iny}
                 \end{subfigure}\hfill
                 \begin{subfigure}{.65\columnwidth} \includegraphics[width=\columnwidth]{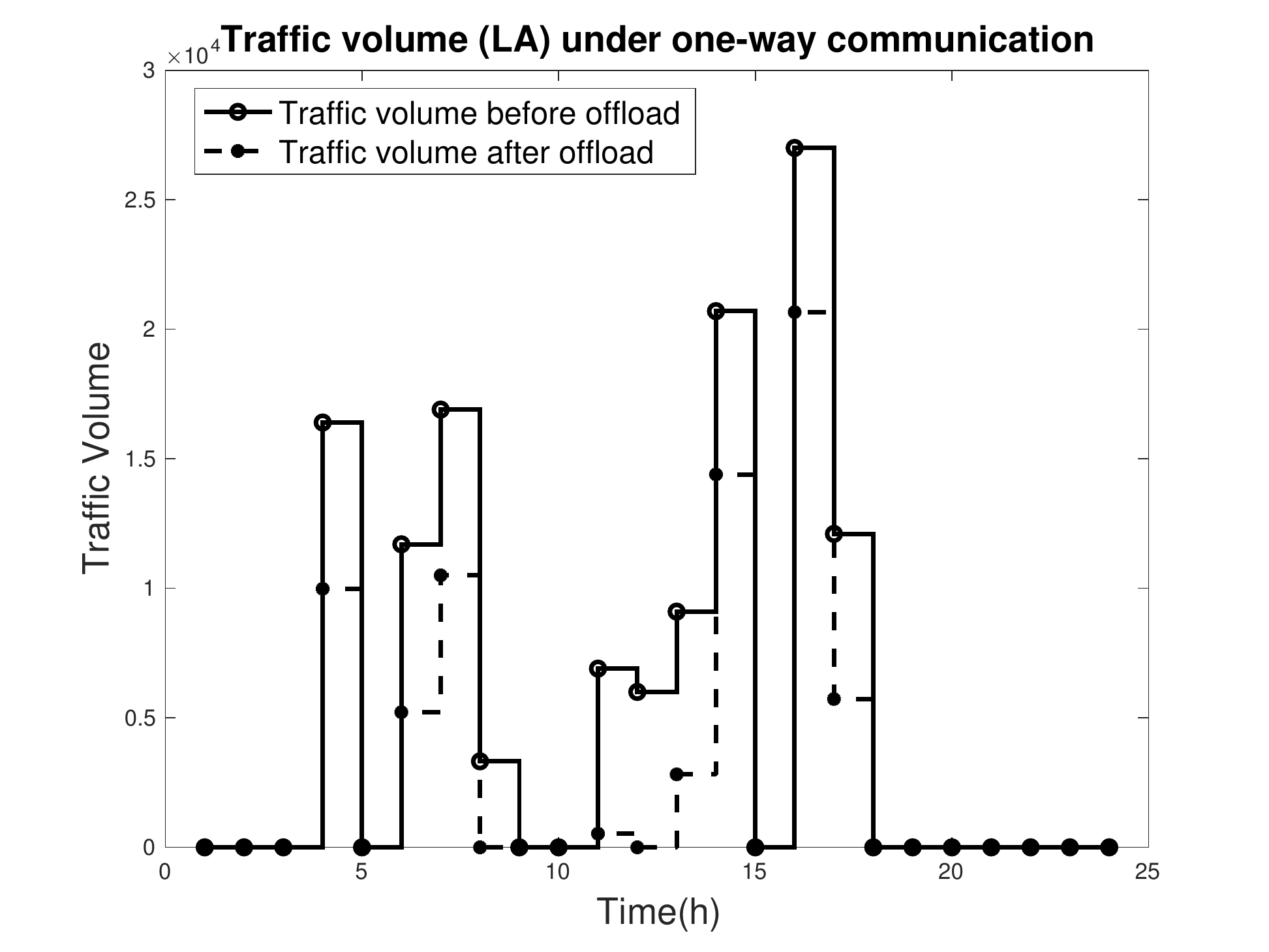}
                 \subcaption {}
                 \label{fig:one way la}
                 \end{subfigure}\hfill
                 \begin{subfigure}{.65\columnwidth} \includegraphics[width=\columnwidth]{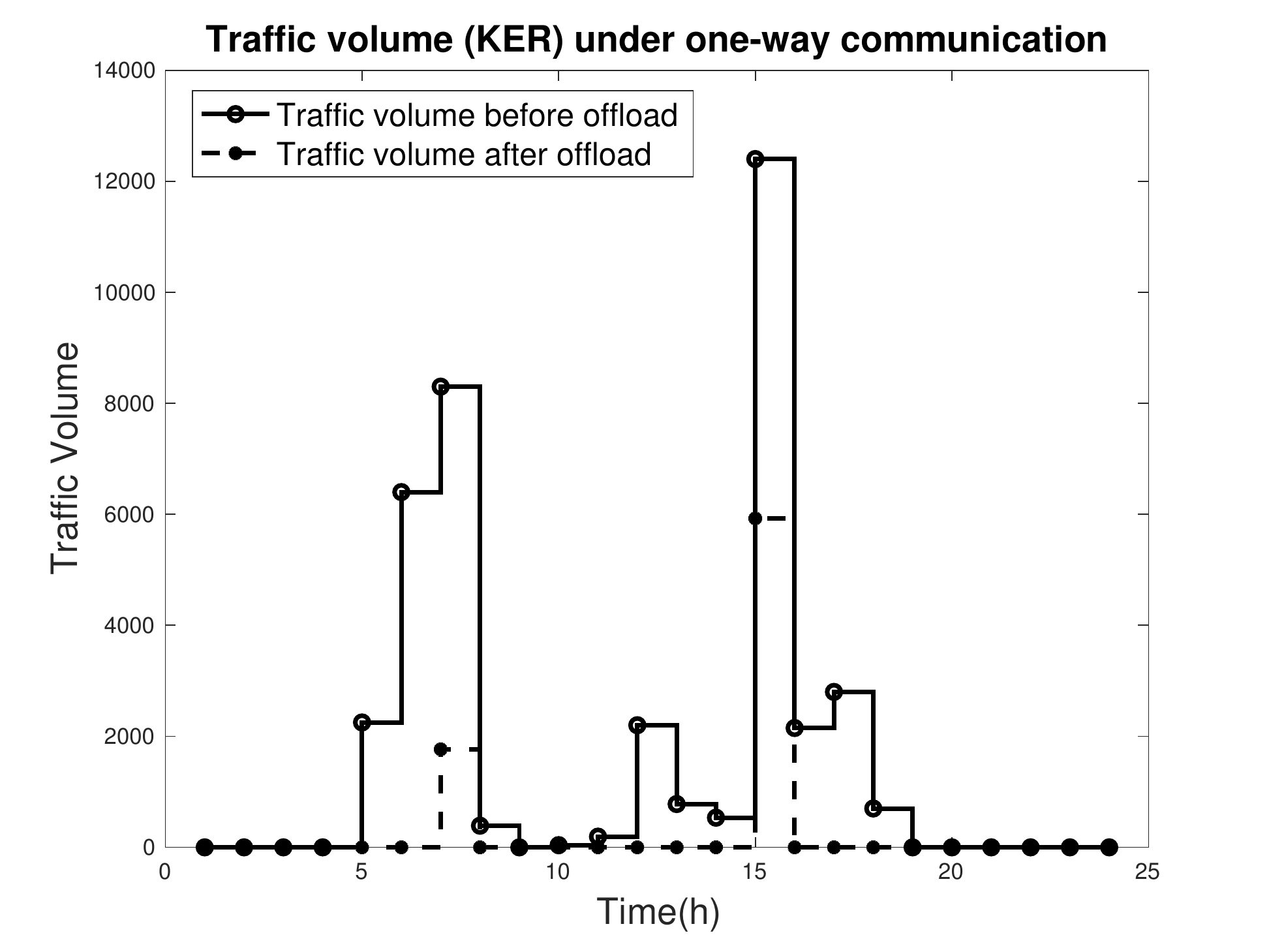}
                 \subcaption{}
                 \label{fig:one way ker}
                 \end{subfigure}\hfill
                 \begin{subfigure}{.65\columnwidth} \includegraphics[width=\columnwidth]{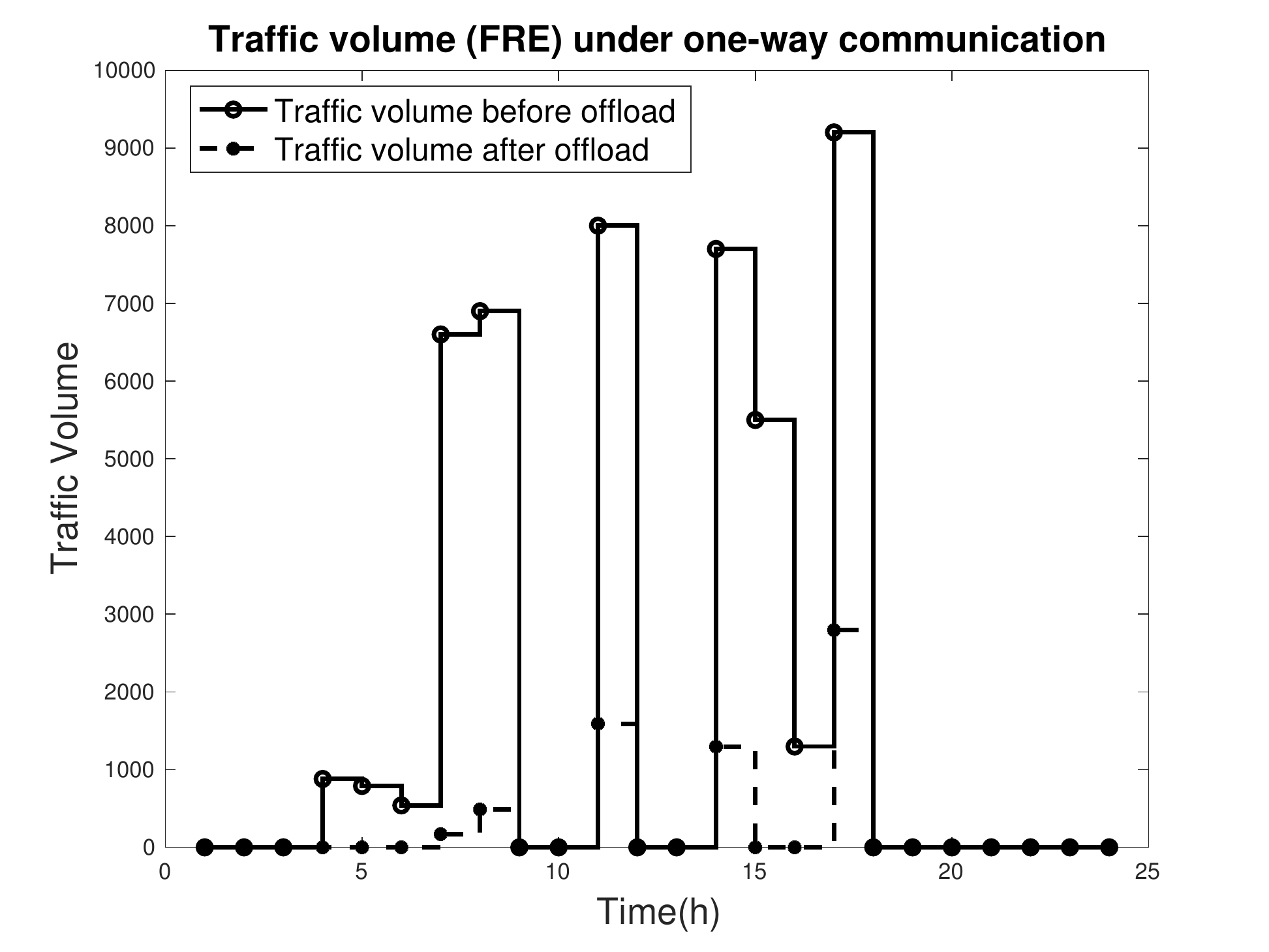}
                 \subcaption{}
                 \label{fig:one way fre}
                 \end{subfigure}\hfill
                 \begin{subfigure}{.65\columnwidth} \includegraphics[width=\columnwidth]{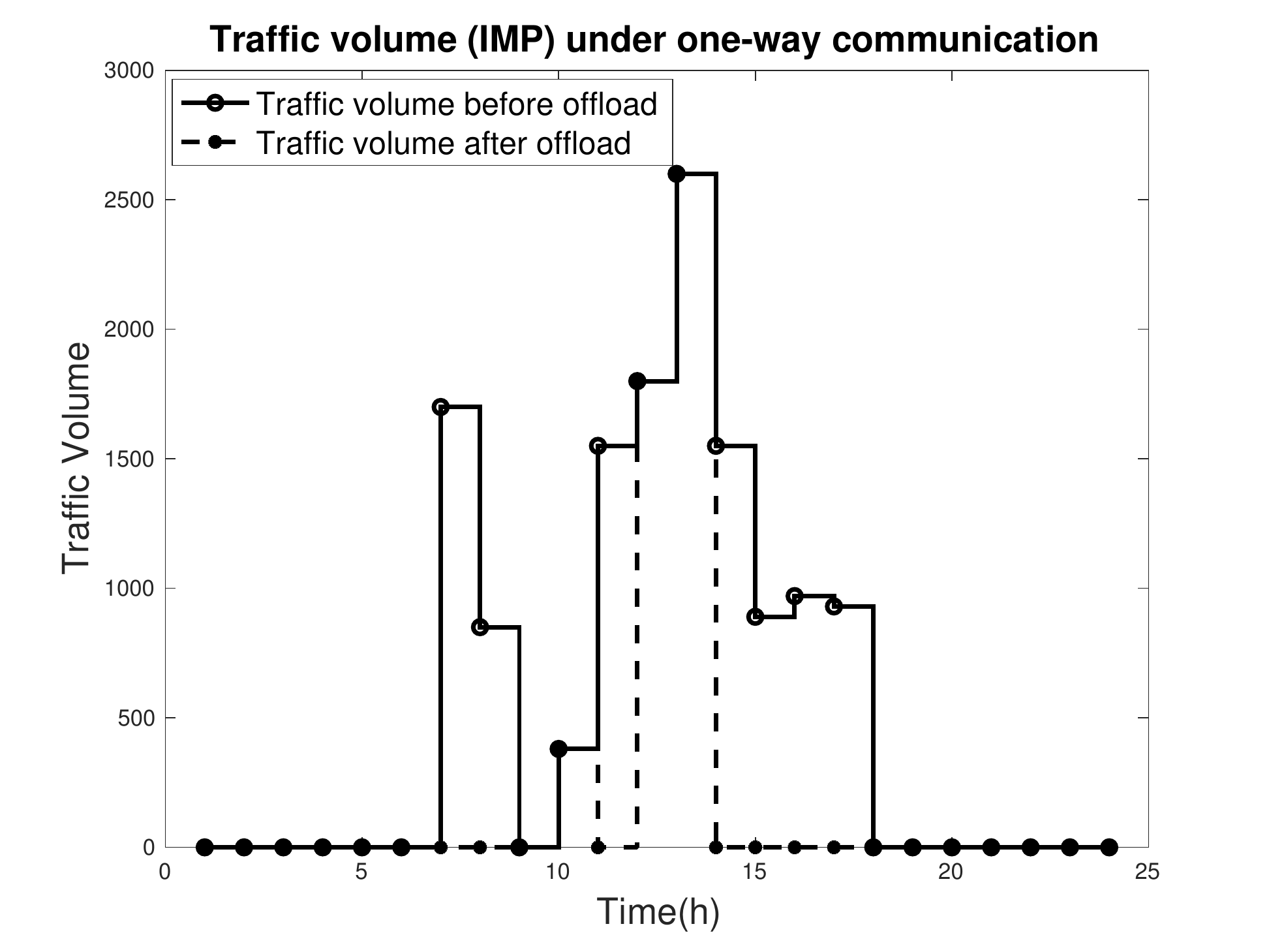}
                 \subcaption{}
                 \label{fig:one way imp}
                 \end{subfigure}\hfill
\caption{In Fig. \ref{fig:one way iny} to Fig. \ref{fig:one way imp}, we present the traffic volume before and after traffic offload. The solid curve is the traffic volume before traffic offload, whereas the dashed curve represents the traffic volume after traffic offload.}
\end{figure*}


\subsection{Case Study Setup}

We consider a government aiming at initiating traffic offload for $S=5$ OD pairs for the next day. Suppose the time horizon $T=24$ and each time slot $t$ is set as $1$ hour. The desired amount of traffic offload at each OD pair is obtained from \cite{PeMS}. We use the peak volume traffic data in $2017$. The $5$ roads that we used in the data set are county `INY' with direction $S$, county `LA' with direction $N$, county `KER' with direction $W$, county `FRE' with direction $S$, and county `IMP' with direction $S$. If a road appears multiple times in the data source, we take the average over the peak volume as the data used in the case study. To show the performance of traffic offload, we use the ahead peak hour traffic volume in \cite{PeMS} as the traffic volume without traffic offload. Since the ahead hourly traffic volume data is not available, we treat the ahead traffic data at different post mile as the traffic volume data at different time.

The size of the passenger set is $N=50000$. We assume the inconvenience cost function $C_{i,s}(q_{i,s})$ of each passenger $i$ is a linear combination of four factors denoted as comfort, reliability, delay on time of arrival, and cost \cite{bhat2006impact,jain2014identifying}. Different passenger assigns different weights on these factors. The weights for each passenger are generated using a multivariate normal distribution, with mean $[0.16,0.27,0.36,0.21]$ and variance $0.3I$, where $I$ is identity matrix \cite{jain2014identifying}.

\subsection{Two-way Communication}
In this section, we demonstrate the proposed approach for the two-way communication scenario. We first generate the passengers' bids. As shown in Theorem \ref{thm:two-way property(efficient)}, the passengers bid truthfully to the government, and hence the government knows the inconvenience cost function of each passenger. The amount of traffic offload that contributed by each passenger is generated using a normal distribution with mean $3.5$ and variance $0.3$. We remark that the contributions model the best effort of all passengers, i.e., the capabilities of all passengers.

We compute the incentives and selection profile following Algorithm \ref{algo:solution}. First, we show the traffic volume on each OD pair before and after traffic offload in Fig. \ref{fig:two way iny} to Fig. \ref{fig:two way imp}. The solid curve is the traffic volume before traffic offload, whereas the dashed curve represents the traffic volume after traffic offload. As observed in Fig. \ref{fig:two way iny} to Fig. \ref{fig:two way imp}, the traffic volume decreases by incentivizing the passengers to switch from private to public transit services. Moreover, the gap between the solid curve and dashed curve gives the amount of traffic offload due to passengers switching from private to public transit services. We next present the social welfare for each OD pair in Fig. \ref{fig:two way iny} to Fig. \ref{fig:two way imp}. We observe that by initiating the traffic offload program, the government earns non-negative social welfare for all time $t$ on each OD pair. We finally give the traffic condition improvement and average payment received by each passenger for each OD pair at $12:00$ PM in Table \ref{table:improvement and payment}. 

\begin{table}[h!]
\centering
\begin{tabular}{|c |c |c |c | c| c|} 
    \hline
    County & INY & LA & FRE & KER & IMP \\
    \hline
    \hline
    Improvement & $7.55\%$ & $90.15\%$ & $0$ & $25.73\%$ & $60.67\%$ \\
    \hline
    Avg. payment & $\$4.164$ & $\$3.8572$ & $\$0$ & $\$3.8683$ & $\$3.7047$\\
    \hline
    \end{tabular}
\caption{Traffic volume improvement and average payment issued to the passengers for each county at $12:00$ PM under two-way communication setting.}
\label{table:improvement and payment}
\end{table}

We present the min-entropy leakage in Fig. \ref{fig:privacy_two_way} to validate that the proposed incentive design in Algorithm \ref{algo:solution} is privacy preserving. We compute the min-entropy for OD pair `INY' at $12:00$ PM when differential privacy parameter $\epsilon$ varies from $0.01$ to $1$. We observe that the min-entropy is monotone increasing with respect to parameter $\epsilon$, which agrees with our privacy preserving property. That is, when the mechanism is designed with stronger privacy guarantee, there exists less min-entropy leakage for each individual passenger.


\begin{figure}[ht]
  \centering
  \begin{subfigure}[b]{0.5\linewidth}
    \centering\includegraphics[width=\linewidth]{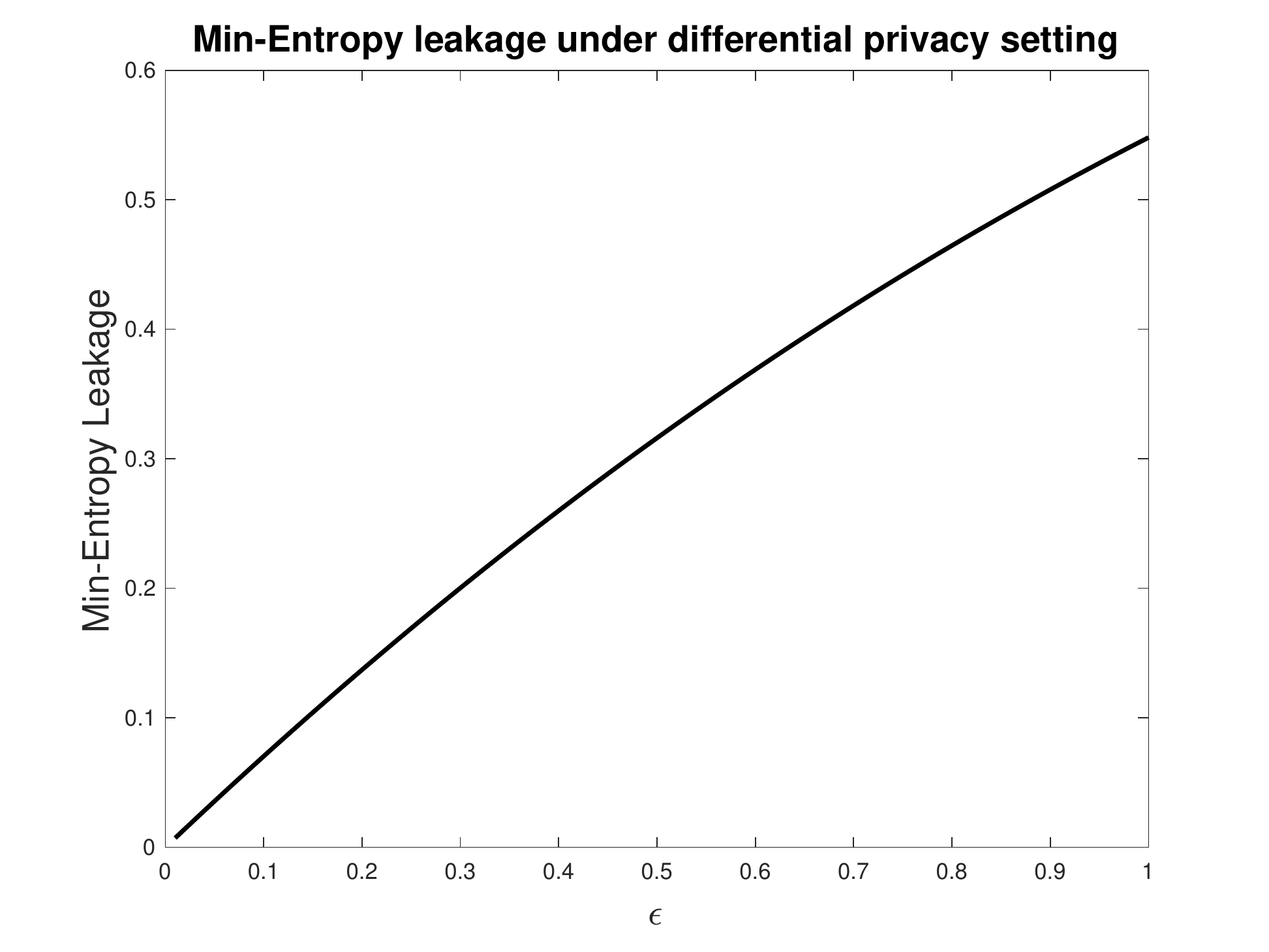}
    \caption{}
    \label{fig:privacy_two_way}
  \end{subfigure}%
  \begin{subfigure}[b]{0.5\linewidth}
    \centering\includegraphics[width=\linewidth]{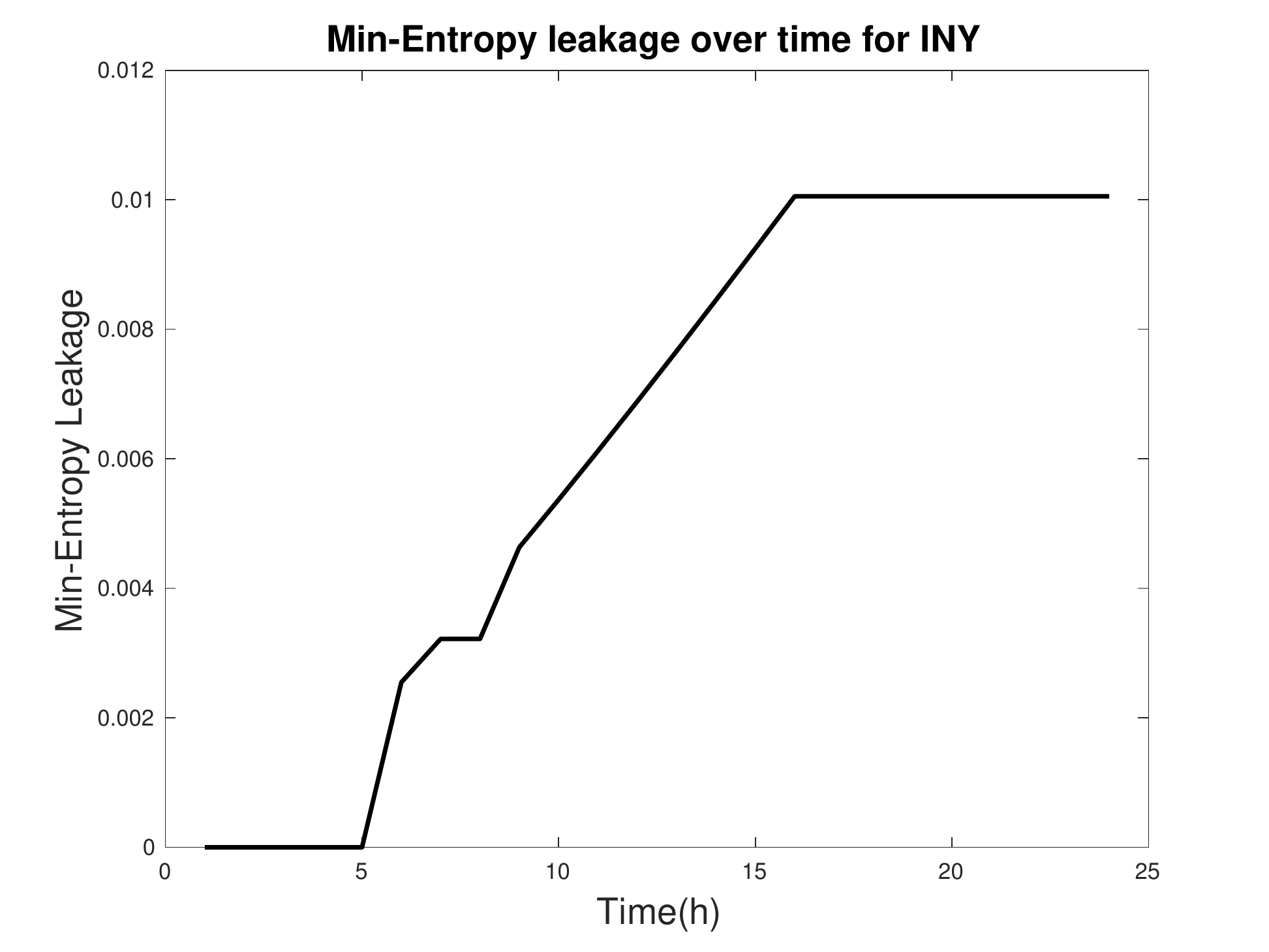}
    \caption{}
    \label{fig:privacy_one_way}
  \end{subfigure}
  \caption{\ref{fig:privacy_two_way} shows min-entropy leakage for OD pair INY at $12:00$ PM when parameter $\epsilon$ varies from $0.01$ to $1$. Fig. \ref{fig:privacy_one_way} shows Min-entropy leakage $L$ for OD pair INY over time.}
\end{figure}

\subsection{One-way Communication}
In this subsection, we demonstrate the proposed approach for the one-way communication scenario. The government initializes a first guess of incentive price $0.02\$$. Given the incentive price $p_{s,t}$, the response from each passenger is computed by Lemma \ref{lemma:passenger response}. The capability of each passenger is adopted from the setting under two-way communication. 

We present the traffic volume on each OD pair before and after traffic offload in Fig. \ref{fig:one way iny} to Fig. \ref{fig:one way imp}. We have the following observations. First, the traffic volume decreases due to passengers switching from private to public transit services. Similar to Fig. \ref{fig:two way iny} to Fig. \ref{fig:two way imp}, the gap between the curves represents the amount of traffic offload. Finally, the traffic volume after traffic offload is lower than that under two-way communication setting for some time $t$, i.e., the amount of traffic offload contributed by the passengers is higher than that under two-way communication setting. The reasons are two fold. First, the government does not know the inconvenience cost function of each passenger under the one-way communication setting and has no ability to select the participating passengers. Therefore, the participating passengers could contribute more than $Q_{s,t}$ for all $s$ and $t$ under one-way communication setting. However, the government selects the winners under two-way communication setting and only $Q_{s,t}$ amount of traffic offload is realized for all $s$ and $t$. Second, the passengers' inconvenience costs are modeled as linear function. Hence any passenger $i$ such that $p_{s,t}\geq C'_{i,s}(q_{i,s})$ would participate in traffic offload by shedding the maximum amount of traffic offload, i.e., contribute its maximum effort. 


We finally present the min-entropy leakage for OD pair (INY) under one-way communication setting in Fig. \ref{fig:privacy_one_way}. In this case study, parameter $\epsilon$ is set as $0.015$. We show how privacy is preserved over time. We observe that the privacy leakage increases over time. The reason is that the malicious party perceives more information over time. Hence, more information can be inferred by the adversary as time increases.

\section{Conclusions}\label{sec:conclusion}
In this paper, we investigate the problem of incentivizing passengers to switch from private to public transit service to mitigate traffic congestion and achieve sustainability. We consider two settings denoted as two-way communication and one-way communication. We model the interaction under former setting using a reverse auction model and propose a polynomial time algorithm to solve for an approximate solution that achieves approximate social optimal, truthfulness, individual rationality, and differential privacy. In the latter setting, we present a convex program to solve for the incentive price. The proposed approach achieves Hannan consistency and differential privacy. The proposed approaches are evaluated using a numerical case study with real-world trace data.

\IEEEpeerreviewmaketitle
\bibliographystyle{IEEEtran}

\vskip -2\baselineskip plus -1fil
\begin{IEEEbiography}[{\includegraphics[width=1in,height=1.25in,clip,keepaspectratio]{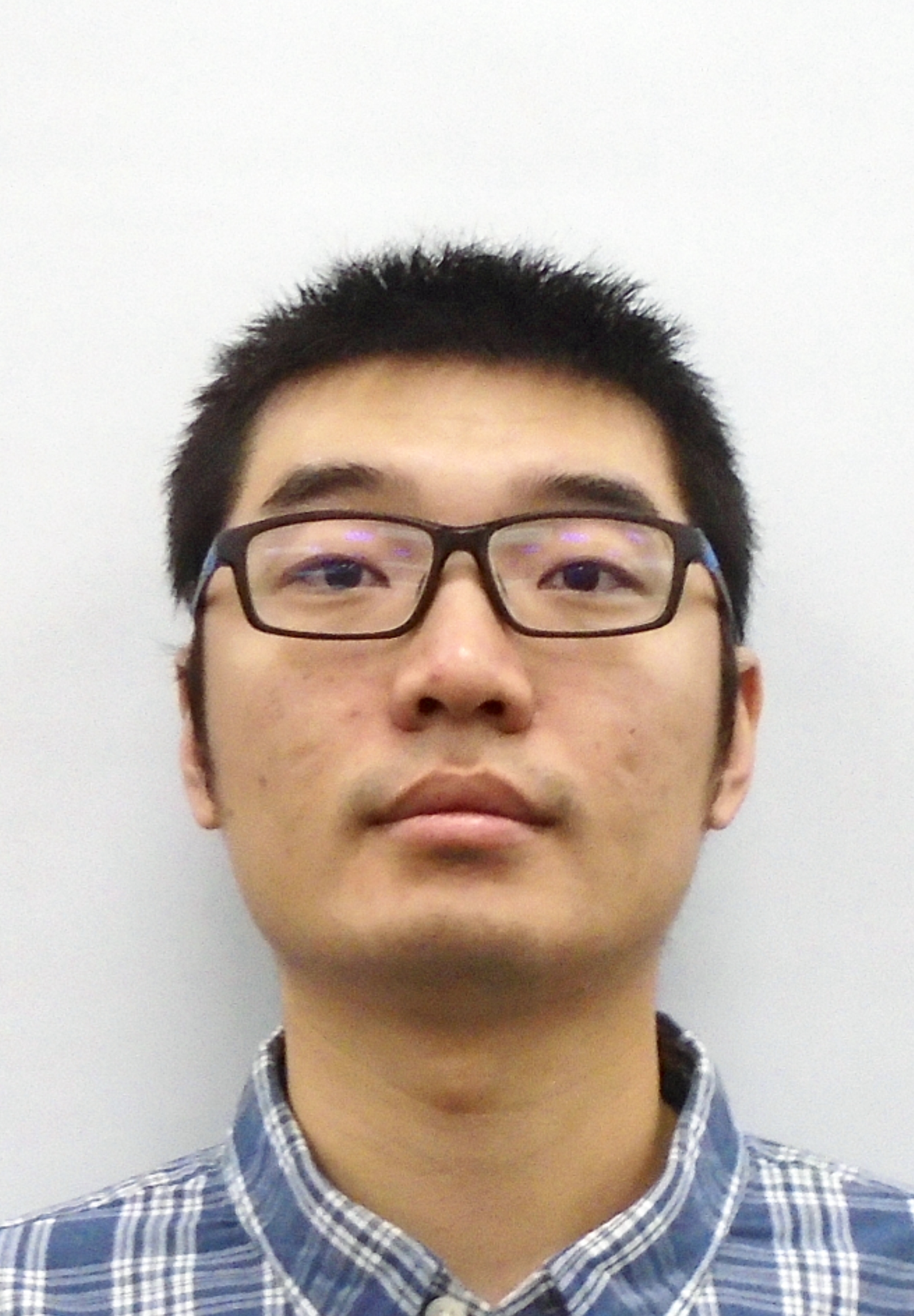}}]
{Luyao Niu}(SM'15)
received the B.Eng. degree from the School of Electro-Mechanical Engineering, Xidian University, Xi’an, China, in 2013 and the M.Sc. degree from the Department of Electrical and Computer Engineering, Worcester Polytechnic Institute (WPI) in 2015. He has been working towards his Ph.D. degree in the Department of Electrical and Computer Engineering at Worcester Polytechnic Institute since 2016. His current research interests include optimization, game theory, and control and security of cyber physical systems. 
\end{IEEEbiography}
\vskip -2.8\baselineskip plus -1fil
\begin{IEEEbiography}[{\includegraphics[width=1in,height=1.25in,clip,keepaspectratio]{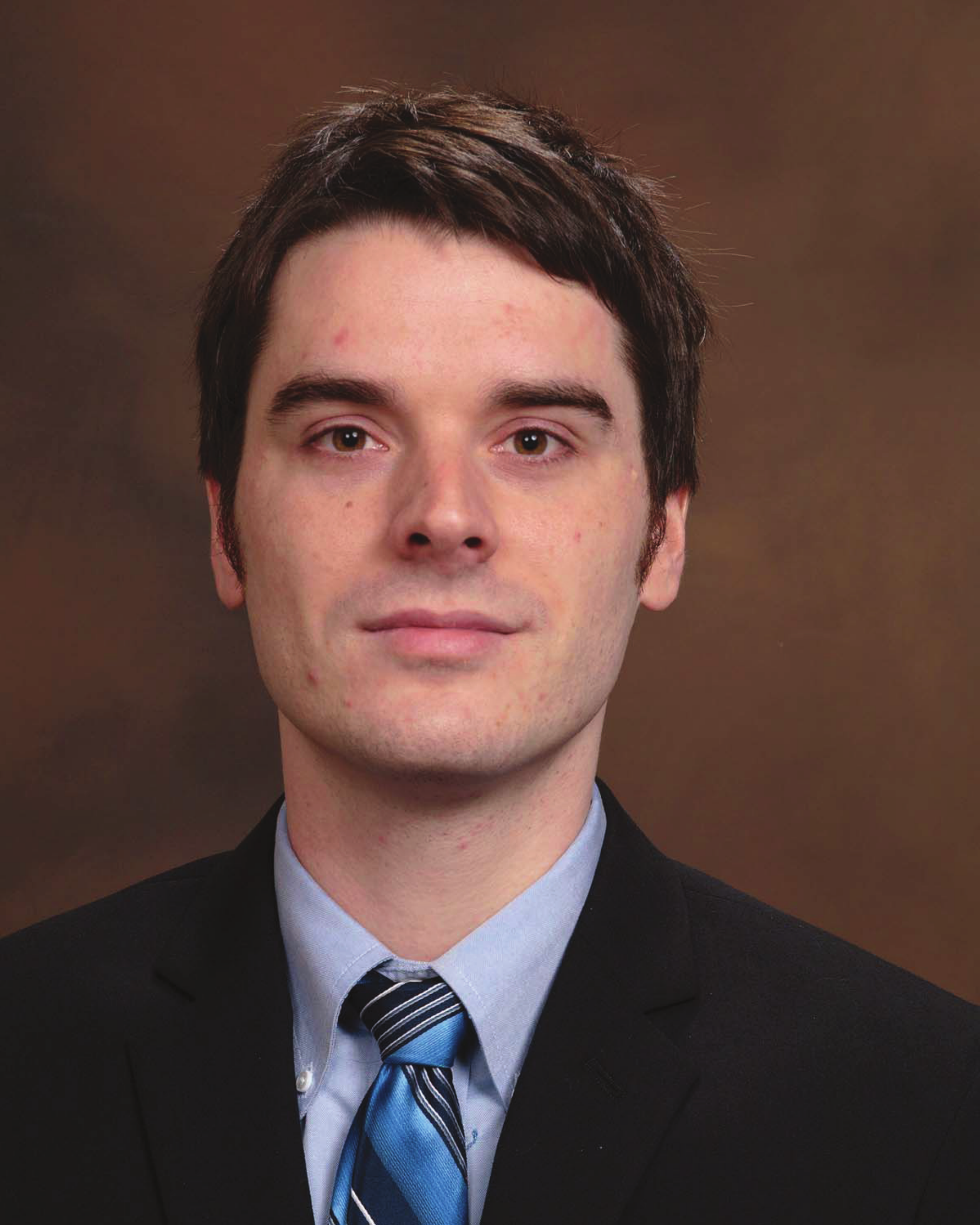}}]{Andrew Clark}(M'15)
is an Assistant Professor in the Department of Electrical and Computer Engineering at Worcester Polytechnic Institute. He received the B.S. degree in Electrical Engineering and the M.S. degree in Mathematics from the University of Michigan - Ann Arbor in 2007 and 2008, respectively. He received the Ph.D. degree from the Network Security Lab (NSL), Department of Electrical Engineering, at the University of Washington - Seattle in 2014. He is author or co-author of the IEEE/IFIP William C. Carter award-winning paper (2010), the WiOpt Best Paper (2012), and the WiOpt Student Best Paper (2014), and was a finalist for the IEEE CDC 2012 Best Student-Paper Award. He received the University of Washington Center for Information Assurance and Cybersecurity (CIAC) Distinguished Research Award (2012) and Distinguished Dissertation Award (2014). His research interests include control and security of complex networks, submodular optimization, and control-theoretic modeling of network security threats. 
\end{IEEEbiography}

\end{document}